\documentclass[format=acmsmall, review=false, nonacm]{acmart}
\usepackage{booktabs}
\usepackage[ruled]{algorithm2e}
\usepackage{hyperref}
\usepackage[noabbrev, nameinlink]{cleveref}
\usepackage{amsmath}
\usepackage{todonotes}
\usepackage{thm-restate}

\SetAlFnt{\small}
\SetAlCapFnt{\small}
\SetAlCapNameFnt{\small}
\SetAlCapHSkip{0pt}
\IncMargin{-\parindent}

\newcommand*{\degree}[1]{\delta(#1)}
\newcommand*{\mindegree}[1]{\delta_{#1}}
\newcommand*{\maxdegree}[1]{\Delta_{#1}}
\newcommand*{\p}[2]{p(\tfrac{#1}{#2})}
\newcommand*{\abs}[1]{\left|#1\right|}
\newcommand*{\simpleset}[1]{\{#1\}}
\newcommand*{\DoI}[1]{\text{DoI}(#1)}
\newcommand*{\col}[1]{c(#1)}
\newcommand*{\union}{\cup}
\newcommand*{\intersect}{\cap}
\newcommand*{\noagent}{\ominus}
\newcommand*{\emptynodes}[1]{\emptyset(#1)}

\newcommand*{\sameColor}[2]{C_{#1}(#2)}
\newcommand*{\neighborhoodSize}[2]{\abs{\neighborhood{#1}{#2}}}
\newcommand*{\edge}[1]{\{#1\}}
\newcommand*{\neighborhood}[2]{N[#1, #2]}
\newcommand*{\set}[2]{\{#1 \mid #2\}}
\newcommand*{\selfInclusiveGame}{game}

\setcitestyle{authoryear}

\title{Single-Peaked Jump Schelling Games}

\author{Tobias Friedrich}
\affiliation{
  \institution{Hasso Plattner Institute, University of Potsdam}
  \city{Potsdam}
  \country{Germany}}
\email{tobias.friedrich@hpi.de}

\author{Pascal Lenzner}
\affiliation{
  \institution{Hasso Plattner Institute, University of Potsdam}
  \city{Potsdam}
  \country{Germany}}
\email{pascal.lenzner@hpi.de}

\author{Louise Molitor}
\affiliation{
  \institution{Hasso Plattner Institute, University of Potsdam}
  \city{Potsdam}
  \country{Germany}}
\email{louise.molitor@hpi.de}

\author{Lars Seifert}
\affiliation{
  \institution{Hasso Plattner Institute, University of Potsdam}
  \city{Potsdam}
  \country{Germany}}
\email{lars.seifert@student.hpi.de}

\begin{abstract}
Schelling games model the wide-spread phenomenon of residential segregation in metropolitan areas from a game-theoretic point of view. In these games agents of different types each strategically select a node on a given graph that models the residential area to maximize their individual utility. The latter solely depends on the types of the agents on neighboring nodes and it has been a standard assumption to consider utility functions that are monotone in the number of same-type neighbors, i.e., more same-type neighbors yield higher utility. This simplifying assumption has recently been challenged since sociological poll results suggest that real-world agents actually favor diverse neighborhoods.

We contribute to the recent endeavor of investigating residential segregation models with realistic agent behavior by studying Jump Schelling Games with agents having a single-peaked utility function. In such games, there are empty nodes in the graph and agents can strategically jump to such nodes to improve their utility. We investigate the existence of equilibria and show that they exist under specific conditions. Contrasting this, we prove that even on simple topologies like paths or rings such stable states are not guaranteed to exist. Regarding the game dynamics, we show that improving response cycles exist independently of the position of the peak in the utility function. Moreover, we show high almost tight bounds on the Price of Anarchy and the Price of Stability with respect to the recently proposed degree of integration, which counts the number of agents with a diverse neighborhood and which serves as a proxy for measuring the segregation strength. Last but not least, we show that computing a beneficial state with high integration is NP-complete and, as a novel conceptual contribution, we also show that it is NP-hard to decide if an equilibrium state can be found via improving response dynamics starting from a given initial state.
\end{abstract}

\begin{document}

\maketitle

\section{Introduction}
Residential segregation~\cite{White86}, i.e., the emergence of regions in metropolitan areas that are homogeneous in terms of ethnicity or socio-economic status of its inhabitants, has been widely studied by social scientists, mathematicians and, recently, also by computer scientists. Segregation has many negative consequences for the inhabitants of a city, for example, it negatively impacts their health~\cite{acevedo2003residential}.

The causes of segregation are complex and range from discriminatory laws to individual action.
Schelling’s classical agent-based model for residential segregation~\cite{Schelling69, Schelling71} specifies a spatial setting where individual agents with a bias towards favoring similar agents care only about the composition of their individual local neighborhoods. This model gives a coherent explanation for the widespread phenomenon of residential segregation, since it shows that local choices by the agents yield globally segregated states~\cite{Clark86, Schelling06}.
In Schelling's model two types of agents, placed on a path and a grid, respectively, act according to the following threshold behavior: agents are \emph{content} with their current position if at least a $\tau$-fraction of neighbors, with $\tau \in (0,1)$, is of their own type. Otherwise, they are discontent and want to move, either via swapping with another random discontent agent or via jumping to an empty position. Starting from a uniformly random distribution, Schelling showed via simulations that the described random process drifts towards strong segregation. This is to be expected if all agents are intolerant, i.e., for $\tau > \tfrac12$. But Schelling's astonishing insight is that this also happens if all agents are tolerant, i.e., for $\tau \leq \tfrac12$.

Many empirical studies in different areas have been conducted to investigate the influence of various parameters on the obtained segregation patterns~{\cite{IEE09, Pan07, Rogers_2011}. In particular, the model has been extensively studied by sociologists~\cite{Bru14, BW07, clark2008understanding} with the help of sophisticated agent-based simulation frameworks such as SimSeg~\cite{fossett1998simseg}. On the theoretical side, the underlying stochastic process leading to segregation was studied~\cite{BEL14, BIK12, BIK17}. Furthermore, Schelling's model recently gained traction within Algorithmic Game Theory, Artificial Intelligence, and Multi-Agent Systems~\cite{A+19, DBLP:journals/aamas/BiloBLM22, BSV21, CIT20, CLM18, BiloBLM22, E+19, KKV21, KKV22}.
	
	Most of these papers are in line with the assumptions made by Schelling and incorporate monotone utility functions, i.e., the agents' utility is monotone in the fraction of same-type neighbors, cf. Figure~\ref{fig:utility}~(left). Although for $\tau < 1$ it is true that no agent prefers segregation locally, agents are equally content in segregated neighborhoods as they are in neighborhoods that just barely meet their tolerance thresholds. However, recent sociological surveys~\cite{gss} show that people actually prefer to live in diverse rather than segregated neighborhoods\footnote{Respondents (on average 78\% white) were asked what they think of “Living in a neighborhood where
		half of your neighbors were blacks?”. A clear majority, e.g. $82\%$ in 2018, responded “strongly favor”, “favor” or “neither favor nor oppose”.}. Based on these observations, different models in which agents prefer integration have been proposed~\cite{Zha04, Zha04b, Pan07}. Very recently Bilò et al.~\cite{BiloBLM22} introduced and analyzed the Single-Peaked Swap Schelling Game, where agents have single-peaked utility functions, cf. Figure~\ref{fig:utility}, and pairs of agents can swap their locations if this is beneficial for both of them.
	
	Based on the model by~\cite{BiloBLM22}, we now take the natural next step and investigate the Jump Schelling Game, where agents can improve their utility by jumping to empty locations, assuming realistic agents having a single-peaked utility function.
	
		\begin{figure}[t]
		\centering
		\includegraphics[width=0.7\linewidth]{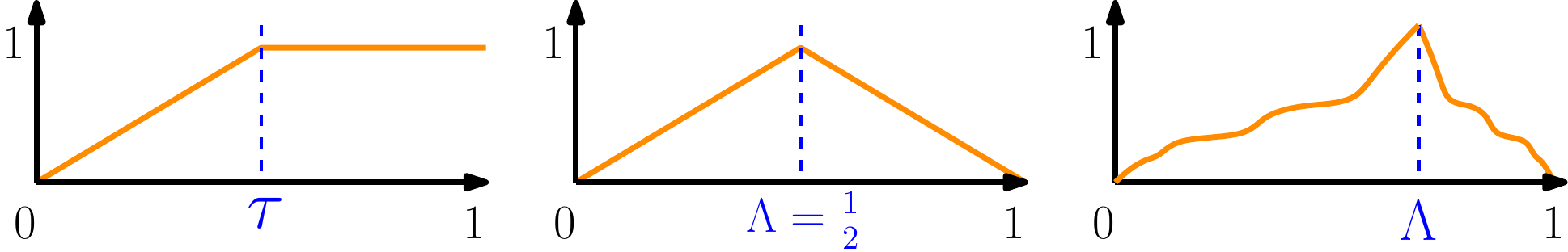}
		\caption{Left: Schelling's original monotone threshold utility function. Middle+Right: single-peaked utility functions. The dashed line marks the utility of an agent if the fraction of same type neighbors meets the threshold and the peak, respectively }
		\label{fig:utility}
	\end{figure}
	
	\paragraph{\textbf{Model}}

	We consider a strategic game played on an undirected, connected graph $G = (V, E)$. For a given node $v \in V$, let $\degree{v}$ be its \textit{degree} and let $\maxdegree{G}$ be the \textit{maximum degree} over all nodes $v \in V$. A graph is \textit{$\delta$-regular}, if $\forall v \in V: \degree{v} = \maxdegree{G}$. We denote with $\alpha(G)$ the \textit{independence number} of $G$, i.e., the cardinality of the maximum independent set in $G$.
	
	A \textit{Single-Peaked Jump Schelling Game} $(G,r,b,\Lambda)$, called the \emph{game}, is defined by a graph $G$, a pair of positive integers with $r \geq 1$ and $1 \leq b \leq r$ and a peak $\Lambda \in (0,1)$. There are two types of agents, which we associate with the colors red and blue. We have the majority type red with $r$ agents and~$b$ blue agents. If $r=b$, we say that the game is \textit{balanced}. For an agent~$i$, let $\col{i}$ be her color.
	
	An agent's \textit{strategy} is her position $v \in V$ on the graph. Each node can only be occupied by at most one agent. The $n = r+b$ strategic agents occupy a strict subset of the nodes in $V$, i.e., there are $e = |V| - n \geq 1$ \textit{empty nodes}.
	A \textit{strategy profile}~$\sigma \in V^n$ is a vector of $n$ distinct nodes in which the $i$-th entry $\sigma(i)$ corresponds to the strategy of the $i$-th agent. We say that an agent $i$ is adjacent to a node $v$ (or an agent $j$) if $G$ has an edge between $\sigma(i)$ and $v$ (resp. $\sigma(j)$).
	For convenience, we use~$\sigma^{-1}$ as a mapping from a node $v \in V$ to the agent occupying~$v$ or~$\noagent$ if~$v$ is empty. The set of empty nodes is $\emptynodes{\sigma} = \set{v \in V}{\sigma^{-1}(v) = \noagent}.$
	
	For an agent~$i$, we define $\sameColor{i}{\sigma} = \set{v \in V \setminus \emptynodes{V}}{\col{\sigma^{-1}(v)} = \col{i}}$ as the set of nodes occupied by agents of the same color in~$\sigma$.
	The \textit{closed neighborhood} of an agent~$i$ in a strategy profile $\sigma$ is $\neighborhood{i}{\sigma} = \{\sigma(i)\} \union \set{v \in V \setminus \emptynodes{\sigma}}{\edge{v,\sigma(i)} \in E}.$
	The agents care about the fraction $f_i(\sigma)$ of agents of their own color, including themselves, in their closed neighborhood where $f_i(\sigma) = \tfrac{
		\abs{\neighborhood{i}{\sigma} \intersect \sameColor{i}{\sigma}}
	}{
		\abs{\neighborhood{i}{\sigma}}
	}.$
	If $f_i(\sigma) = 1$, we say that agent $i$ is \textit{segregated}. Furthermore, observe that we have $f_i(\sigma) > 0$ for any agent~$i$, since $\sigma(i) \in \neighborhood{i}{\sigma}$. Also, we emphasize that our definition of $f_i(\sigma)$ deviates from similar definitions in related work. In particular, \citet{CLM18,E+19}, and \citet{A+19} exclude the respective agent $i$ from her neighborhood, while \citet{KKV21} count agent $i$ only in the denominator of $f_i(\sigma)$. The different existing definitions of the homogeneity of a neighborhood all have their individual strengths and weaknesses. We decided to follow the definition of \citet{BiloBLM22}, since they established the swap variant of our model. The key idea of their definition is that agents contribute to the diversity of their neighborhood. Thus, agents actively strive for integration. We think that this best captures the single-peaked setting.

	The \textit{utility} of an agent~$i$ is $U_i(\sigma) = p\left(f_i(\sigma)\right)$, with $p$ being an arbitrary single-peaked function with peak $\Lambda \in (0,1)$ and the following properties: (1) $p(0) = 0$ and $p(x)$ is strictly monotonically increasing on $[0, \Lambda]$, (2) for all $ x \in [\Lambda, 1]$ it holds that $p(x) = \p{\Lambda (1-x)}{1 - \Lambda}$.
	W.l.o.g., we further assume that $p(\Lambda) = 1$. See Figure~\ref{fig:utility}~(middle+right) for an illustration. Note that we explicitly exclude $\Lambda = 1$ and from our definition it follows that $p(1) = 0$. Allowing $\Lambda = 1$ would allow monotone utilities, similar to the models in \citep{CLM18,A+19}, where agents actively strive for segregation or passively accept it. However, we assume that agents actively strive for diversity. Thus, a completely homogeneous neighborhood should not be acceptable. This justifies $p(1)=0$. Hence, both $\Lambda<1$ and $p(1)=0$ model integration-oriented agents and go hand in hand.

	The strategic agents attempt to choose their strategy to maximize their utility. The only way in which an agent can change her strategy is to \textit{jump}, i.e., to choose an empty node $v \in \emptynodes{\sigma}$ as her new location. We denote the resulting strategy profile after a jump of agent $i$ to a node $v$ as $\sigma_{iv}$. A jump is \textit{improving}, if $U_i(\sigma) < U_i(\sigma_{iv})$.
	A strategy profile $\sigma$ is a (pure) \textit{Nash Equilibrium} (NE) if and only if there are no improving jumps, i.e., for all agents~$i$ and nodes~$v \in \emptynodes{\sigma}$, we have $U_i(\sigma) \geq U_i(\sigma_{iv})$.
	
	A measure to quantify the amount of segregation in a strategy profile $\sigma$ is the \textit{degree of integration} (DoI), which counts the number of non-segregated agents, hence $\DoI{\sigma} = \abs{\set{i}{f_i(\sigma) < 1}}$.
	For a game $(G,r,b,\Lambda)$, let $\sigma^*$ be a strategy profile that maximizes the DoI and let $\text{NE}(G,r,b,\Lambda)$ be its set of Nash Equilibria. We evaluate the impact of the agents' selfishness on the overall social welfare by studying the \textit{Price of Anarchy} (PoA), defined as
	$\text{PoA}(G,r,b,\Lambda) = \tfrac{\DoI{\sigma^*}}{\min_{\sigma \in NE(G,r,b,\Lambda)} \DoI{\sigma}}$
	and the \textit{Price of Stability} (PoS),  defined as
	$\text{PoS}(G,r,b,\Lambda) = \tfrac{\DoI{\sigma^*}}{\max_{\sigma \in NE(G,r,b,\Lambda)} \DoI{\sigma}}.$
	If the best (resp. worst) NE has a DoI of 0, the PoS (resp. PoA) is unbounded.
	
	A game has the \textit{finite improvement property} (FIP) if and only if, starting from any strategy profile~$\sigma$, the game will always reach a NE in a finite number of steps. As proven by \cite{MS96}, this is equivalent to the game being a \textit{generalized ordinal potential game}. In particular, the FIP does not hold if there is a cycle of strategy profiles $\sigma^0, \sigma^1, \dots, \sigma^k = \sigma^0$, such that for any $k' < k$, there is an agent $i$ and empty node $v \in \emptynodes{\sigma^{k'}}$ with $\sigma^{k'+1} = \sigma^{k'}_{iv}$ and $U_i(\sigma^{k'}) < U_i(\sigma^{k'+1})$. These cycles are known as \textit{improving response cycles} (IRCs).
	
	\paragraph{\textbf{Related Work.}}
	Game-theoretic models for residential segregation were first studied by~\cite{CLM18} and later extended by~\cite{E+19}. There, agents have a monotone utility function as shown in Figure~\ref{fig:utility} (left). Additionally, agents may also have location preferences. The authors study the FIP and the PoA in terms of the number of content agents. 
	\cite{A+19} consider a simplified model using the most extreme monotone threshold-based utility function with $\tau=1$. They prove results on the existence of equilibria, in particular, that equilibria are not guaranteed to exist on trees, and on the complexity of deciding equilibrium existence. Also, they introduce the DoI as social welfare measure and they study the PoA in terms of utilitarian social welfare and in terms of the DoI. For the latter, they obtain a tight bound of $\tfrac{n}{2}$ on the PoA and the PoS that is achieved on a tree. In contrast, on paths, they derive a constant PoS. The complexity results were extended by~\cite{KBFN22}, in particular, they show that deciding the existence of NE in the swap version as well as in the jump version of the simplified model is NP-hard. \cite{DBLP:journals/aamas/BiloBLM22} strengthened the PoA results for the swap version w.r.t. the utilitarian social welfare function and investigated the model on almost regular graphs, grids and paths. Additionally, they introduce a variant with locality. \cite{CIT20} studied a variant of the Jump Schelling Game with $\tau = 1$ where the agents’ utility is a function of the composition of their neighborhood and of the social influence by agents that select the same location. \cite{KKV22}  considered a generalized variant, where an ordering of the agent types exists and agents are more tolerant towards agents of types that are closer according to the ordering. Another novel variant of the Jump Schelling Game was investigated by~\cite{KKV21} . There the main new aspect is that an agent is included when counting her neighborhood size. This subtle change leads to agents preferring locations with more own-type neighbors. \cite{BSV21} measure social welfare via the number of agents with non-zero utility, they prove hardness results for computing the social optimum and discuss other solution concepts, like Pareto optimality.    
	
	Most related is the recent work by \cite{BiloBLM22}, which studies the same model as we do, but there only pairs of agents can improve their utility by swapping their locations. They find that equilibria are not guaranteed to exist in general, but they do exist for $\Lambda = \tfrac12$ on bipartite graphs and for $\Lambda \leq \tfrac12$ on almost regular graphs. The latter is shown via an ordinal potential function, i.e., convergence of IRDs is guaranteed. For the PoA they prove an upper bound of $\min\{\Delta(G),\tfrac{n}{b+1}\}$ and give almost tight lower bounds for bipartite graphs and regular graphs. Also, they lower bound the PoS by $\Omega(\sqrt{n\Lambda})$ and give constant bounds on bipartite and almost regular graphs. Note that due to the existence of empty nodes in our model, our results cannot be directly compared. 
	
	Also related are hedonic diversity games~\cite{BEI19,BE20,GHKSS22} where selfish agents form coalitions and the utility of an agent only depends on the type distribution of her coalition. For such games, single-peaked utility functions yield favorable game-theoretic properties.
	
	\paragraph{\textbf{Our Contribution}}
	We investigate Jump Schelling Games with agents having a single-peaked utility function. In contrast to monotone utility functions that have been studied in earlier work, this assumption better reflects recent sociological poll results on real-world agent behavior~\cite{gss}. Moreover, this transition to a different type of utility function is also interesting from a technical point of view since it yields insights into the properties of Schelling-type systems under different preconditions.
	
	Regarding the existence of pure NE, we provide a collection of positive and negative results. On the negative side, we show that NE are not guaranteed to exist on the simplest possible topologies, i.e., on paths and rings with single-peaked utilities with peak at least~$\tfrac12$.
	Note that this is in contrast to the version with monotone utilities where for the case of rings NE always exist.
	On the positive side, we give various conditions that enable NE existence, e.g., such states are guaranteed to exist if the underlying graph has a sufficiently large independent set, 
	or if it has sufficiently many degree~$1$ nodes. 
	The situation is worse for the convergence of game dynamics. We show that even on regular graphs IRCs exist independently of the position of the peak in the utility function. Moreover, this even holds for the special case with a peak at $\tfrac12$ and only a single empty node. These negative results for $\Lambda \leq \tfrac12$ also represent a marked contrast to the swap version, where convergence is guaranteed for this case on almost regular graphs.

	With regard to the quality of the equilibria, we focus on the DoI as social cost function. This measure has gained popularity since it can be understood as a simple proxy for the obtained segregation strength as it counts the number of agents having close contact with some agent of a different type. For the PoA with respect to the DoI, we establish that the technique for deriving an upper bound for single-peaked Swap Schelling Games can be adapted to also work in our setting. This yields the same PoA upper bound of $\min\{\Delta(G),n/(b+1)\}$. Subsequently, we give almost matching PoA lower bounds and we prove that also the lower bounds for the PoS almost match this high upper bound. On the positive side, we show that on graphs with a sufficiently large independent set, the PoS depends on the ratio of the largest and the smallest node degree in the graph, which implies for this case a PoS of $1$ on regular graphs that also holds for rings with a single empty node.
	
	Last but not least, we consider complexity aspects of our model. Analogously to previous work on the Jump Schelling Game with monotone utilities and to work on Swap Schelling Games with single-peaked utilities, we focus on the hardness of computing a strategy profile with a high degree of integration. Using a novel technique relying on the \textsc{Max SAT} problem, we show that this problem is NP-complete, improving on an earlier result by \cite{A+19}. Moreover, as a novel conceptual contribution, we investigate the hardness of finding an equilibrium state via improving response dynamics. As one of our main results, we show that this problem is NP-hard. So far, researchers have studied the complexity of deciding the existence of an equilibrium for a given instance of a Schelling Game. We depart from this, since even if it can be decided efficiently that for some instance an equilibrium exists, guiding the agents towards this equilibrium from a given initial state is complicated, since this would involve a potentially very complex centrally coordinated relocation of many agents in a single step. In contrast, reaching an equilibrium via a sequence of improving moves is much easier to coordinate, since in every step the respective move can be recommended and, since this is an improving move, the agents will follow this advice. 
	
	Overall we find that making the model more realistic by employing single-peaked utilities entails a significantly different behavior of the model compared to the variant with monotone utilities but also compared to Single-Peaked Swap Schelling Games.

	\section{Game Dynamics}\label{sec:GameDynamics}
	In this section we show that even on very simple graph classes improving response dynamics are not guaranteed to converge to stable states. Moreover, we provide IRCs for the entire range of~$\Lambda$.
	Note, that given an IRC for a game on a graph~$G$ IRCs exist for all games on any graph $H$ that contains $G$ as a node-induced subgraph since we can add empty nodes to $G$ to obtain $H$ without interfering with the IRC.
	We start with an IRC for $\Lambda \geq \tfrac{1}{2}$. 
	\begin{restatable}{theorem}{theoremone}\label{theorem:self_inclusive_irc_rings} 
		For $\Lambda \geq \tfrac{1}{2}$, the game violates the FIP even on rings and paths with $e \geq 2$.
	\end{restatable}
	\begin{proof}
		Consider a game with five nodes, two red agents and one blue agent on a ring or path. We start with a strategy profile in which the blue agent is adjacent to both red agents. An illustration is given in \Cref{fig:ring_irc}. As $\Lambda \geq \tfrac{1}{2}$, the blue agent prefers to be in a neighborhood with only one of the red agents. Hence, an improving jump from the blue agent results in one segregated red agent. As a consequence the red agent jumps to the node adjacent to the blue agent.
		Further, observe that at no point in this cycle does any other agent have an improving jump and none of the jumping agents have an alternative improving jump (except for symmetry). 
	\end{proof}
	\begin{figure}[t]
	\centering
	\includegraphics[width = 0.45\textwidth]{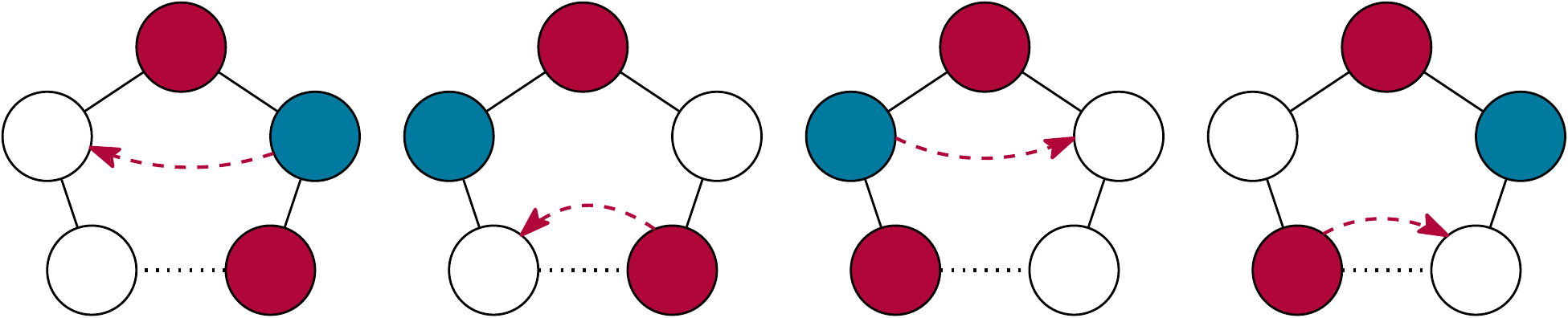}
	\caption{IRC on a ring (path without the dotted edge) for $\Lambda \geq \tfrac{1}{2}$. }
	\label{fig:ring_irc}
\end{figure}
	
	We now show that IRCs also exist for games with agents who prefer to be in the minority. 
	\begin{restatable}{theorem}{theoremtwo}
		For $\Lambda \leq \tfrac{1}{2}$, the game violates the FIP even on regular graphs.
		\label{theorem:self_inclusive_irc_less_than_1_2}
	\end{restatable}
	\begin{proof}
		Consider \Cref{fig:irc_lambda_less_than_1_2}. The graph~$G$ has $\maxdegree{G} = 7$. The red agent with utility $p\left(\frac35\right)$ can improve to $p\left(\frac12\right)$. Afterwards, the two previously adjacent blue agents are segregated and jump as well. This causes the utility of the three blue agents in the third row to drop to $p\left(\frac35\right)$. By jumping to the lower part, they can improve to $p\left(\frac12\right)$. The resulting strategy profile is identical to the first one.
	\end{proof}
	
	\begin{figure}[h]
		\centering
		\includegraphics[width = 0.6\textwidth]{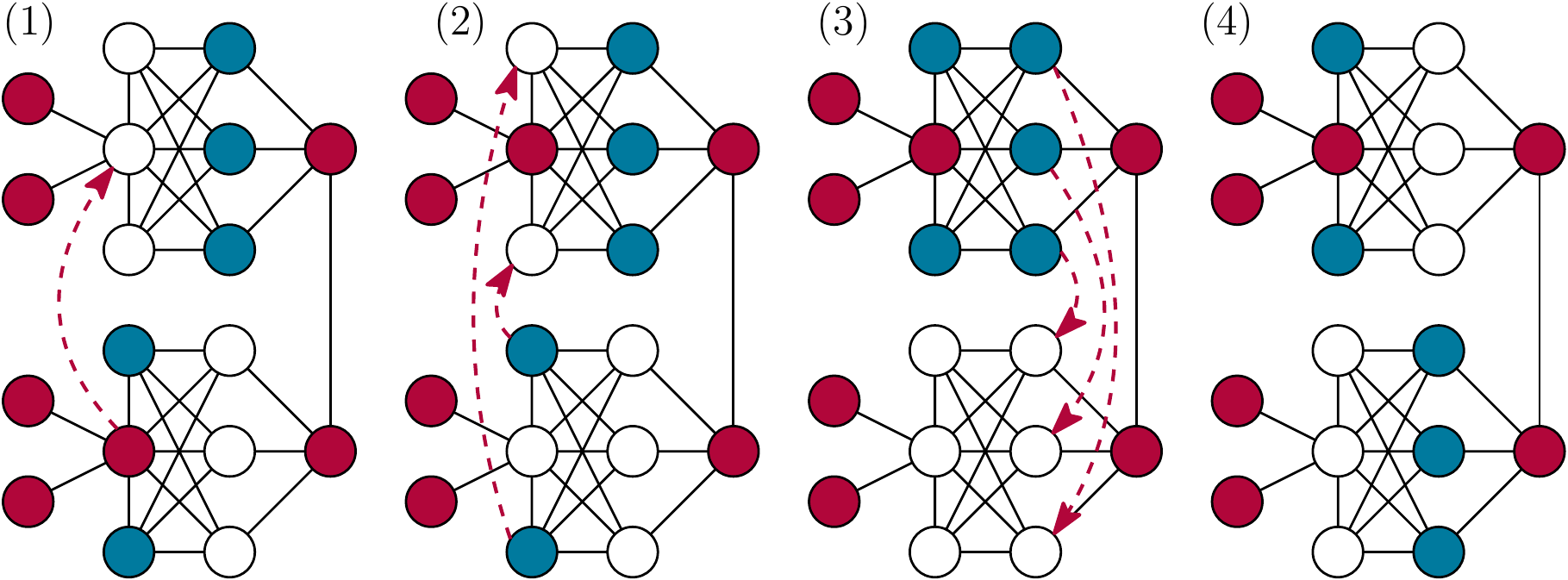}
		\caption{IRC for $\Lambda \leq \tfrac{1}{2}$. An IRC for a regular graph can be obtained by adding empty nodes.}
		\label{fig:irc_lambda_less_than_1_2}
	\end{figure}
	
	Next, we have that even for the special case with $\Lambda = \frac{1}{2}$ and only one single empty node no convergence is guaranteed.
	\begin{restatable}{theorem}{theoremthree}
		For $\Lambda = \tfrac{1}{2}$, the game violates the FIP even on regular graphs with $e=1$.
		\label{theorem:self_inclusive_irc_one_empty}
	\end{restatable}
\begin{figure}[t]
	\centering
	\includegraphics[width = \textwidth]{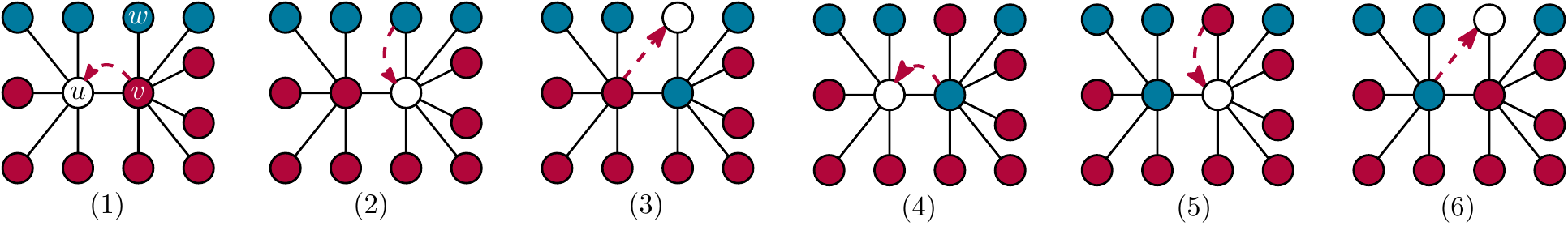}
	\caption{An IRC for a \selfInclusiveGame~with $\Lambda = \tfrac{1}{2}$ and $e=1$.}
	\label{fig:irc_self_inclusive_one_empty}
\end{figure}

\begin{proof}
	\Cref{fig:irc_self_inclusive_one_empty} shows an IRC for an instance with only one empty node. In the initial strategy profile, we have a pair of adjacent nodes $u$ and $v$, one adjacent to three red and three blue agents and one adjacent to five red and two blue agents.
	In the first step, the red agent~$i$ on $v$ with a utility of~$p\left(\frac57\right)$ performs an improving jump to $u$ to get a utility of~$p\left(\frac46\right)$. Thereby, the blue agent $j$ on a node~$w$ gets segregated and jumps to $v$. Yet, the new utility of~$i$ on~$w$ is~$1$, while the current utility of~$i$ is merely~$p\left(\frac47\right)$. Then,~$j$ jumps to the former position of agent~$i$, as~$p\left(\frac12\right) > p\left(\frac27\right)$. Thereby, agent~$i$ gets segregated and returns to her original position $v$. Finally, the utility of $j$ on her original position $w$ is~$p\left(\frac12\right)$, clearly better than her current utility of~$p\left(\frac37\right)$ on $u$.
\end{proof}

On the positive side, we can show for a very special case on rings that convergence is guaranteed.
\begin{restatable}{theorem}{theoremnine}
	On rings, the \selfInclusiveGame~with $e=1$ and $\Lambda = \frac{1}{2}$ is an ordinal potential game. It converges after at most $n$ steps.
	\label{theorem:rings_self_inclusive_e1_convergence}
\end{restatable}

\begin{proof}
	We claim that for each improving jump of an agent $i$ to a node~$v$, we have $\DoI{\sigma_{iv}} \geq \DoI{\sigma} +1$. Hence,~$\DoI{\sigma}$ is an ordinal potential function and a NE must be reached after at most~$n$ improving jumps.
	
	Assume there is an agent~$i$ with an improving jump to~$v$, i.e., $U_i(\sigma) < U_i(\sigma_{iv})$.
	We claim $U_i(\sigma) = 0$. Assume $U_i(\sigma) > 0$, i.e., either $p\left(\frac12\right)=1$ or $p\left(\frac13\right) = p\left(\frac23\right)$.  In the first case, agent~$i$ already has the highest possible utility and thus no incentive to jump. In the second case $\left(U_i(\sigma) = p\left(\frac13\right)\right)$, we must have $U_i(\sigma_{iv})=1$. But since~$v$ is the only empty node this is only possible if $\sigma(i)$ and $v$ are adjacent. However, this requires $\neighborhoodSize{i}{\sigma} =2 \neq 3$. 
	
	Therefore, in $\sigma$, agent~$i$ is not adjacent to any agent of the other color and in~$\sigma_{iv}$ adjacent to at least one agent of the other color. 
	Thus, any agent adjacent to~$i$ that has a utility larger than~0 in~$\sigma$ still has a utility larger than~$0$ in~$\sigma_{iv}$. Also, no agent adjacent to~$v$ can drop to utility~$0$ because of~$i$ jumping to $v$. Thus, we have $\DoI{\sigma}+1 \leq  \DoI{\sigma_{iv}}$. 
\end{proof}

	\section{Existence of Equilibria}
	A fundamental question is if NE always exist. We start with a negative result that even on rings existence of equilibria is not guaranteed for $\Lambda \geq \tfrac{1}{2}$. However, in certain cases, we can provide existential results. In particular, equilibria exist if the underlying graph has an independent set that is large enough or if the graph contains sufficiently many leaf nodes. Moreover, for regular graphs, we show that equilibria exist if $e=1$ and $r$ is large enough.
	The following non-existence result for rings follows from \Cref{theorem:self_inclusive_irc_rings}.
	
	\begin{corollary}
		Even on rings, the existence of equilibria for the \selfInclusiveGame~is not guaranteed for $\Lambda \geq \tfrac{1}{2}$.
		\label{corollary:rings_without_equilibria}
	\end{corollary}
	\begin{proof}
		Consider the instance in \Cref{fig:ring_irc}. Clearly, in a NE the red agents must be adjacent to the blue agent. Moreover, the IRC starts with a strategy profile in which the blue agent is adjacent to both red agents. Therefore, no equilibria can exist.
	\end{proof}
	
	If the independence number is at least the number of blue agents plus the number of empty nodes, existence of NE is guaranteed. This result is similar to the swap version~\cite{BiloBLM22}.
	\begin{theorem}
		Every \selfInclusiveGame~
		on a graph with an independent set of size $\alpha(G) \geq b + e$ has a NE.
		\label{theorem:independent_set}
	\end{theorem}
	\begin{proof}
		Let $I$ be the nodes of an independent set of size $k = b+e$. We construct a NE~$\sigma$. To this end, we first place all red agents on $V \setminus I$. Note that regardless of how we distribute the blue agents on~$I$, no red agent wants to jump to an empty node of~$I$ as any red agent has a utility of $0$ there.
		Observe that if we place a blue agent on a node $v \in I$, she has a utility of $p\left(\frac{1}{\degree{v}+1}\right)$, no matter where the other blue agents are placed. We order the nodes $v \in I$ in descending order by $p\left(\frac{1}{\degree{v}+1}\right)$ and place the $b$ blue agents on the~$b$ nodes with the highest utility. Thus, no blue agent has an incentive to jump to another empty node as by our placement her assigned location is at least equally good. Hence, the strategy profile $\sigma$ is a NE.
	\end{proof}
	
	Thus, if $r$ is large enough NE always exist on bipartite graphs.
	
	\begin{corollary}
		Every \selfInclusiveGame~ 
        with $r \geq \tfrac{\abs{V}}{2}$ played on a bipartite graph admits a NE that can be computed efficiently.
		\label{corollary:je_existence_bipartite}
	\end{corollary}
	
	Next, we show that for $\Lambda \geq \tfrac12$ games with a low number of empty nodes and a low difference between the number of red and blue agents proportional to the number of empty nodes admit a NE. To this end, we consider a special kind of independent sets.
	
	\begin{definition}
		A \textit{maximum degree independent set (max-deg IS)} is an independent set $I$, such that $\forall u \in I, v \in V \setminus I: \degree{v} \leq \degree{u}$. The size of the largest max-deg IS of a graph $G$ is $\alpha^{\max \delta}(G)$.
	\end{definition}
	Note that for any graph, it holds that $\alpha^{\max \delta}(G) \geq 1$.
	
		\begin{figure}[t]
		\centering
		\includegraphics[width=0.3\textwidth]{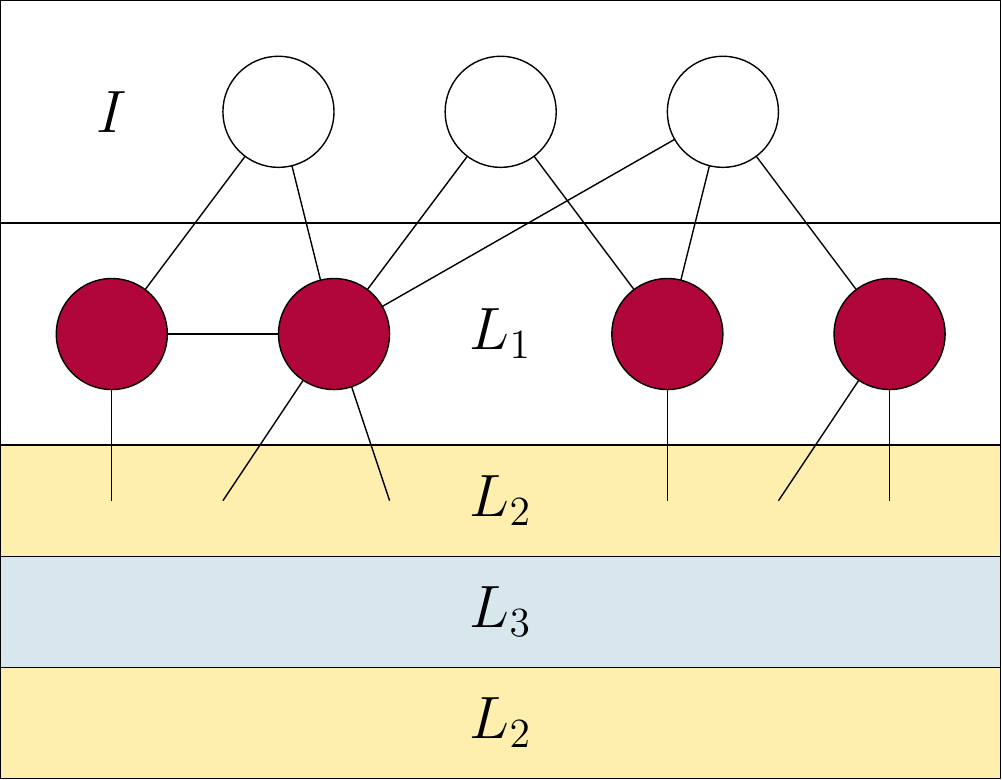}
		\caption{The layer graph. We first have an independent set~$I$ of~$e$ nodes, then a layer $L_1$ of at most $\maxdegree{G}\cdot e$ red agents. The nodes in the following layers are either part of~$L_2$ (for even layers) or $L_3$ (for odd layers).}
		\label{fig:self_inclusive_layer_graph}
	\end{figure}
	\begin{theorem}
		Let $G$ be a graph with $e \leq \alpha^{\max \delta}(G)$ and $e \leq \tfrac{r-b}{\maxdegree{G}}$. For $\Lambda \geq \tfrac{1}{2}$, the \selfInclusiveGame~
		has a NE.
		\label{theorem:self_inclusive_NE_existence_max_degree_IS}
	\end{theorem}
	\begin{proof}
		Let $I$ be a max-deg IS of size $e$. Since $e\leq \alpha^{\max \delta}(G)$ this exists. We place red agents on all nodes adjacent to nodes in~$I$. For this, we need at most $\maxdegree{G} \cdot e$ red agents. Afterward, we have~$r'$ red agents left and $r' \geq r - \maxdegree{G} \cdot e \geq b$.
		
		We claim that we can place the remaining agents on the remaining nodes, such that every blue agent is adjacent to at least one red agent.
		For this, consider the layer graph rooted at an imaginary node that results from merging all~$e$ nodes in~$I$, cf. \Cref{fig:self_inclusive_layer_graph}. Let the root layer be layer~$0$. Note that therefore, layer~$1$ is fully occupied by the~$r-r'$ red agents we placed in the first step on nodes adjacent to nodes in~$I$. Let~$L_2$ be the set of nodes in all even layers (except for layer~0) and~$L_3$ be the set of nodes in all odd layers (except for layer~1). 
		
		Note that all nodes in~$L_2$ (resp.~$L_3$) have at least one adjacent node not in~$L_2$ (resp. $L_3$). Furthermore, we have $|L_2| + |L_3| = r' + b$. Hence,~$|L_2|$ or~$|L_3|$ is at least $\tfrac{r'+b}{2}$. Since $r' \geq b$, it follows $b \leq \tfrac{r'+b}{2}$, so there is $L \in \simpleset{L_2, L_3}$ with $|L| \geq b$. We place all blue nodes in~$L$ and all red nodes on the remaining empty spots in~$L_2$,~$L_3$. Then, every blue node has at least one red neighbor.
		
		The placement~$\sigma$ is stable.
		As all empty nodes are adjacent to only red nodes, no red agent wants to jump.
		Let~$i$ be a blue agent and~$u$ be an empty node. By construction, $\sigma(i) \not\in I$ and $u \in I$. At least one neighbor of~$\sigma(i)$ is red, hence~$i$ has a non-zero utility.
		Since $\Lambda \geq \tfrac{1}{2}$, the worst non-zero utility is $p\left(\frac{1}{\degree{\sigma(i)}+1}\right)$.
		Thus, $U_i(\sigma) \geq p\left(\frac{1}{\degree{\sigma(i)}+1}\right)$ and since all neighbors of~$u$ are red, $U_i(\sigma_{iu}) = p\left(\frac{1}{\degree{u}+1}\right)$.
		As~$I$ is a max-deg IS, we have $\degree{\sigma(i)} \leq \degree{u}$. Furthermore, it follows from $\Lambda \geq \tfrac{1}{2}$ that $p\left(\frac{1}{\degree{u}+1}\right) \leq p\left(\frac{1}{\degree{\sigma(i)}+1}\right) = U_i(\sigma)$. Hence,~$i$ has no improving jump.
	\end{proof}

	Note that for regular graphs any independent set is a max-deg IS, i.e., $\alpha^{\max \delta}(G) = \alpha(G) \geq \tfrac{\abs{V}}{\delta+1}$. 
	\begin{corollary}
		Any \selfInclusiveGame~on a $\delta$-regular graph $G$ with $e \leq \alpha(G)$, $r \geq b + \delta \cdot e$ and $\Lambda \geq \tfrac{1}{2}$ has NE.
	\end{corollary}
	
	Next, we show that graphs with a large number of leaves admit NE. In particular, this applies to trees with many leaves, e.g., stars.
	\begin{theorem}
		Every \selfInclusiveGame~
		with $\Lambda \geq \tfrac{1}{2}$ on a graph with at least~$b$ nodes of degree one admits NE.\label{theorem:existence_many_leaves}
	\end{theorem}
	\begin{proof}
		Since $r,b,e \geq 1$ the nodes of degree one are not adjacent to each other. We place all blue agents on degree one nodes. Let~$R$ be the set of nodes adjacent to blue nodes. We have $\abs{R} \leq b$, and since $r \geq b$, we can place red agents on all of them. The remaining red agents can be placed anywhere.
		With this, no empty node is adjacent to a blue agent. Thus, no red agent has an improving jump. Furthermore, each blue agent~$i$ has a utility of $ p\left(\frac{1}{2}\right)$ and since $\Lambda \geq \tfrac{1}{2}$, for any empty node~$v$, $U_i(\sigma_{iv}) \leq  p\left(\frac{1}{2}\right)$ holds. Hence,~$\sigma$ is a NE.
	\end{proof}
	
	While even for regular
	graphs with $e =1$ the FIP is violated, we can guarantee the existence of NE with further conditions.
	
	\begin{theorem}
		For any \selfInclusiveGame~on a $\delta$-regular graph with $\Lambda \geq \tfrac{1}{2}$, $r \geq \delta$ and $e=1$, equilibria exist and can be computed efficiently.	\label{theorem:self_inclusive_NE_existence_r_geq_delta_e1}
	\end{theorem}
	\begin{proof}
		Consider a strategy profile $\sigma$ in which the only empty node $v$ is surrounded by~$\delta$ red agents.
		For any agent~$i$, it holds that if~$i$ is red, $U_i(\sigma_{iv}) = 0$ and if~$i$ is blue, we have that $U_i(\sigma_{iv}) =  p\left(\frac{1}{\delta+1}\right)$. Since $\Lambda \geq \tfrac{1}{2}$, there is no smaller, non-zero utility than $ p\left(\frac{1}{\delta+1}\right)$, therefore any agent with an improving jump to~$v$ must be blue and must have~0 utility in~$\sigma$, i.e., be segregated.
		If there is a segregated agent~$i$ with $U_i(\sigma_{iv})=0$, then after the jump, the new empty node~$\sigma(i)$ has a blue monochromatic neighborhood. Thus, in any strategy profile~$\sigma'$, reached from~$\sigma$ through improving response dynamics, the empty node has again a monochromatic neighborhood.
		Consider $\text{DoI}(\sigma)$, i.e., the number of agents with positive utility, in~$\sigma$. Given an agent~$i$ with an improving jump to~$v$, it holds that $i$ is segregated in~$\sigma$ and non-segregated in~$\sigma_{iv}$. Furthermore, since all neighbors of~$i$ in~$\sigma$ are of her color and all new neighbors in~$\sigma_{iv}$ of the other color, no new segregated agent is created. Hence, $\text{DoI}(\sigma_{iv}) > \text{DoI}(\sigma)$.
		As the DoI is upper bounded by~$n$, we get that starting from~$\sigma$, there can only be a finite number of improving jumps before an equilibrium is reached.
	\end{proof}
	
	\section{Price of Anarchy and Stability}
	\label{section:self_inclusive_price_of_anarchy_stability}
	In this section, we study the PoA and PoS of the \selfInclusiveGame~with respect to the DoI. We already showed that the existence of equilibria is not guaranteed for many instances, yet, we still give bounds that apply whenever equilibria do exist.
	
	\subsection{\textbf{Price of Anarchy}}
	We start with the PoA.
	The next lemma provides a necessary condition that holds for any NE.
	
	\begin{lemma}
		No NE contains segregated agents of different colors.
		\label{lemma:self_inclusive_only_one_color_segregated}
	\end{lemma}
	\begin{proof}
		Assume towards a contradiction, that~$\sigma$ is a NE with two segregated agents~$i$ and~$j$ and $\col{i} \neq \col{j}$. Without loss of generality let~$i$ be red and~$j$ be blue. Let~$v$ be an empty node adjacent to some agent~$k$; since~$G$ is connected such a node must exist. Then, if~$k$ is red (resp. blue), agent~$j$ (resp.~$i$) has a profitable jump to~$v$, so~$\sigma$ cannot be a NE.
	\end{proof}
	As shown in \cite{BiloBLM22} (Lemma~5, Theorem~5), \Cref{lemma:self_inclusive_only_one_color_segregated} can be used to get a bound on the PoA for the swap version. The proofs do not rely on swaps and thus carry over.
	\begin{lemma}
		For any \selfInclusiveGame~$(G,r,b,\Lambda)$ and strategy profile~$\sigma$, we have $\text{DoI}(\sigma) \leq \min((\maxdegree{G}+1)b, n)$.
		\label{lemma:self_inclusive_poa_doi_upper_bound}
	\end{lemma}

With this, we get the same upper bound as in~\cite{BiloBLM22}. 
	
	\begin{theorem}[\cite{BiloBLM22}]
		For any \selfInclusiveGame, $\text{PoA}(G,r,b,\Lambda) \leq \min \left(\maxdegree{G}, \tfrac{n}{b+1} \right)$.\label{thm:singlepeaked}
	\end{theorem}
	
	It still remains to be shown that this upper bound is tight. We show that this is, asymptotically with respect to $\maxdegree{G}$, the case for general graphs.
	\begin{theorem}
		For any\ $\Lambda$, there exists a \selfInclusiveGame~$(G, r,b, \Lambda)$ with $\text{PoA}(G,r,b,\Lambda) = \tfrac{n}{b+1} = \maxdegree{G}-1$.
		\label{theorem:self_inclusive_general_poa_lower_bound}
	\end{theorem}
	\begin{proof}
		For some $\delta \geq 4$, consider the game $(G,r,b,\Lambda)$ with $b = \delta-1, r = b^2$ depicted in \Cref{fig:PoA_worst_case}. The graph~$G$ has a node~$v$ adjacent to a set~$B$ of~$b$ nodes. Further,~$v$ is adjacent to another node, which lies on a path of altogether $b$ nodes, which at the same time represent the root of a tree. Hence, each node on this path is adjacent to one node in $B'$, each of which is adjacent to~$b$ nodes in total. Observe that $\maxdegree{G}=\delta$ and there are $b+1$ empty nodes.
		
		There is an optimal strategy profile~$\sigma^*$ in which all nodes in $B'$ are occupied by blue agents and all nodes outside of $B \union B' \union \simpleset{v}$ are occupied by red agents. We have that $\text{DoI}(\sigma^*) = \maxdegree{G}(\maxdegree{G}-1)=n$.
		Furthermore, there is a NE~$\sigma$ in which the blue agents occupy $B$ and~$b+1$ of the leaf nodes adjacent to nodes in $B'$ are empty. Since each blue agent is adjacent to exactly one red agent, we have for any red agent~$i$ and empty node~$u$ that $U_i(\sigma) = U_i(\sigma_{iu}) = p\left(\frac{1}{2}\right)$. Thus, we have that~$\sigma$ is a NE. Only the blue agents and one red agent are not segregated, hence it holds that $\text{DoI}(\sigma) = b+1 = \maxdegree{G}$.
		With this we have that $\text{PoA}(G,r,b, \Lambda) \geq \tfrac{n}{b+1} = \tfrac{\maxdegree{G}(\maxdegree{G}-1)}{\maxdegree{G}}$.
	\end{proof}
	\begin{figure}[t]
		\centering
		\includegraphics[width = 0.95\textwidth]{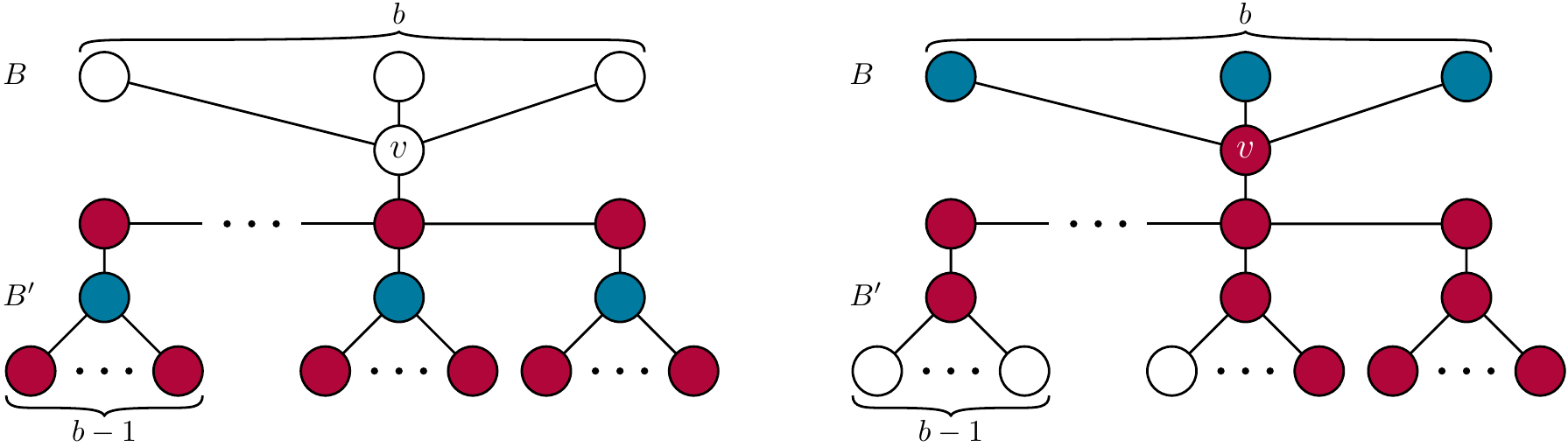}
		\caption{A \selfInclusiveGame~
		with $b = \maxdegree{G}-1$ and $r = b^2$. The middle row is a path of $b$ many nodes occupied by red agents.  The dots in the lower row are representative for the rest of the $b-1$ many leaf nodes. Left: Optimum $\sigma^*$ with $\text{DoI}(\sigma^*) = \maxdegree{G}(\maxdegree{G}-1) = n$. Right: NE $\sigma$ with $\text{DoI}(\sigma) = b+1 = \maxdegree{G}$.}
		\label{fig:PoA_worst_case}
	\end{figure}
	We use a similar construction as in~\cite{BiloBLM22} to also obtain a lower bound for a regular graph. Yet, in our case the bound holds for all values of $\Lambda$ instead for only $\Lambda < \frac12$~\cite{BiloBLM22}.

	\begin{figure}[b]
		\centering
		\includegraphics[width = 0.75\textwidth]{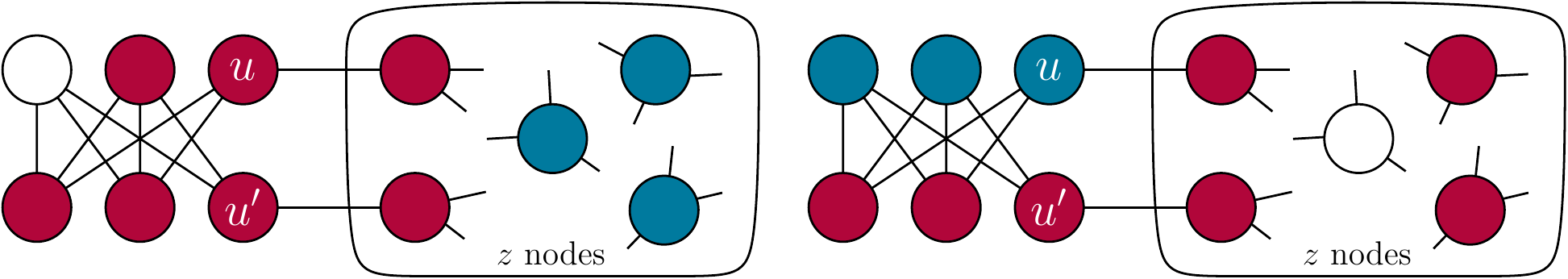}
		\caption{$(G,r,\delta,\Lambda)$ on a $\delta$-regular graph with $z \geq \delta^2+1$. Left: Optimum $\sigma^*$ with $\text{DoI}(\sigma^*)=\delta(\delta+1)$. The nodes of the blue agents and the empty node form an independent set. Right: NE $\sigma$ with $\text{DoI}(\sigma) = \delta+1$.}
		\label{fig:si_poa_delta_regular_lower_bound}
	\end{figure}
	
	\begin{theorem}
		For every $\delta \geq 2$ and $\Lambda$, a \selfInclusiveGame~$(G, r,b, \Lambda)$ on a $\delta$-regular graph with $\text{PoA}(G,r,b,\Lambda) \geq \tfrac{\delta(\delta+1)}{2\delta+1} = \tfrac{\delta+1}{2} - \tfrac{\delta+1}{4\delta+2}$ exists.
		\label{theorem:self_inclusive_poa_delta_regular_lower_bound}
	\end{theorem}
	\begin{proof}
		Consider~\Cref{fig:si_poa_delta_regular_lower_bound}. For a fixed $\delta$, consider the game $(G,r,b,\Lambda)$ in which~$G$ is a $\delta$-regular graph consisting of a left and a right gadget. The left gadget is a $K_{\delta, \delta}$ from which the edge between two nodes $u,u'$ has been removed. The right gadget consists of $z \geq \delta^2+1$ nodes that are connected in some arbitrary way such that~$G$ is a $\delta$-regular graph. The two gadgets are connected via $u, u'$. Let $b = \delta$ and $r = \delta+z-1$ and therefore~$e=1$. Since~$G$ is a $\delta$-regular graph with $\abs{V} \geq 2 \delta + \delta^2+1 = (\delta+1)^2$, it follows that there must exist an independent set~$I$ of size $\abs{I} \geq \tfrac{(\delta+1)^2}{\delta+1} = \delta+1$. 
		
		Consider the strategy profile~$\sigma^*$ in which all red agents are~placed on nodes outside of~$I$. Every blue agent is adjacent to~$\delta$ red agents. Hence, $\text{DoI}(\sigma^*) = \delta(\delta+1)$.
		Yet, there is a NE~$\sigma$ in which the blue agents occupy the upper half of the $K_{\delta,\delta}$ gadget and the empty node~$v$ is not adjacent to a blue agent. Clearly, no red agent wants to jump and for every blue agent $i$, it holds $U_i(\sigma) = U_i(\sigma_{iv})$. We have $\text{DoI}(\sigma) = 2\delta+1$.
		Thus, it holds that $\text{PoA}(G,r,b,\Lambda) \geq \tfrac{\delta(\delta+1)}{2\delta+1} = \tfrac{\delta+1}{2} - \tfrac{\delta+1}{4\delta+2}$.
	\end{proof}

	\begin{figure}[t]
		\centering
		\includegraphics[width=0.6\textwidth]{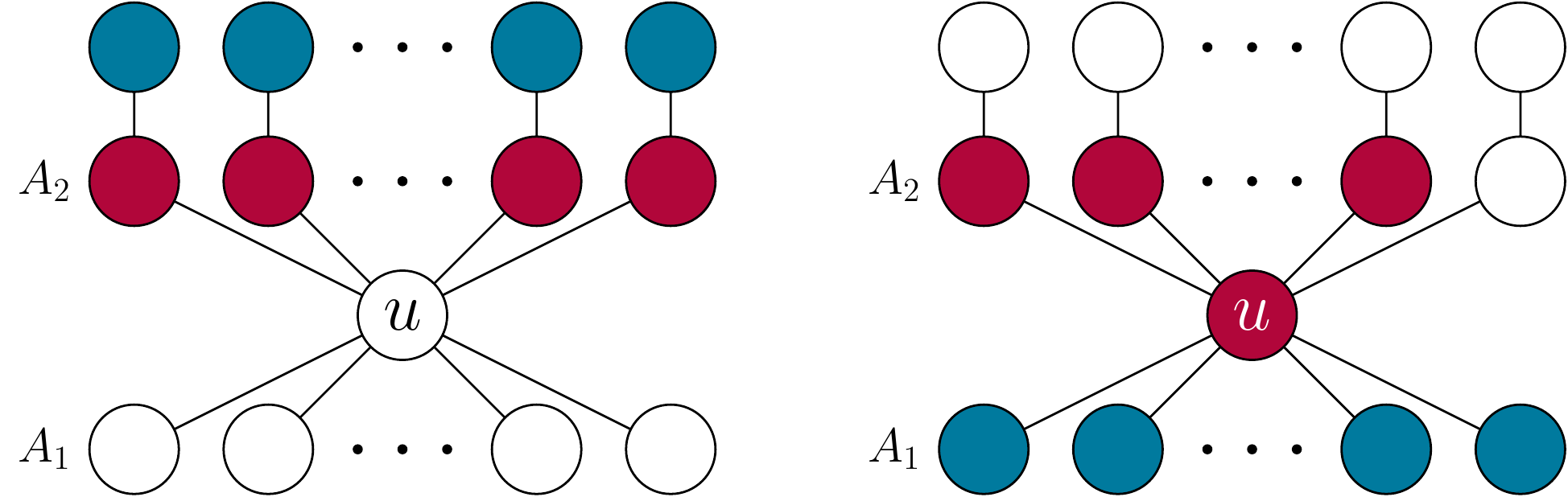}
		\caption{Balanced \selfInclusiveGame~with $b$ agents per type. Left: Optimum $\sigma^*$: $\text{DoI}(\sigma^*)=n$. Right: NE $\sigma$ with $\text{DoI}(\sigma) = b+1$.}
		\label{fig:balanced_si_game_poa}
	\end{figure}
	For games with $r=b$, it follows from~\Cref{thm:singlepeaked} that the PoA is at most $\frac{2b}{b+1} < 2$. We show that this bound is tight.
	\begin{theorem}
		For any $\Lambda$ and $b$, there is a balanced \selfInclusiveGame~$(G,b,b,\Lambda)$ with $\text{PoA}(G,b,b,\Lambda) = \tfrac{2b}{b+1}$.
		\label{theorem:self_inclusive_poa_balanced_lower_bound_general}
	\end{theorem}
	\begin{proof}
		Consider the balanced game $(G,b,b,\Lambda)$ in \Cref{fig:balanced_si_game_poa} in which the graph~$G$ has a node~$u$ adjacent to two sets $A_1,A_2$ of~$b$ nodes each. All nodes in~$A_1$ are leaves and all nodes in~$A_2$ are adjacent to one leaf each. Thus,~$G$ has $3b+1$ nodes in total.
		
		In the optimal strategy profile, all red agents are placed on~$A_2$ and all blue agents on the leafs are adjacent to the nodes of~$A_2$. Thus, all~$2b$ agents are non-segregated and $\text{DoI}(\sigma) = n$.
		
		However, there is a NE~$\sigma$ in which~$u$ is occupied by a red agent, all blue agents are located on nodes in~$A_1$ and the other red agents are on nodes of~$A_2$. No red agent has an improving jump, as no empty node is adjacent to blue agents. Furthermore, all blue agents~$i$ have $U_i(\sigma) = p\left(\frac{1}{2}\right)$. Observe that all empty nodes are adjacent to at most one red agent, and therefore, for any empty node~$v$, we have $U_i(\sigma_{iv}) = p\left(\frac{1}{2}\right)$ or $U_i(\sigma_{iv})=0$. Hence, it holds that~$\sigma$ is a~NE.
		The red agent on~$u$ is the only non-segregated red agent, thus $\text{DoI}(\sigma) = b+1$, and we have a PoA of $\tfrac{2b}{b+1}$.
	\end{proof}
	
	\subsection{\textbf{Price of Stability.}}
	\label{section:self_inclusive_price_of_stability}
	We now study the PoS and give bounds under different conditions. First, we observe from \Cref{thm:singlepeaked}, that for any game $(G,r,b,\Lambda)$, we have
	$\text{PoS}(G,r,b,\Lambda) \leq \min \left( \maxdegree{G}, \tfrac{n}{b+1} \right).$
	We now present a lower bound which, although only for~$b=1$, is asymptotically tight.
	
	\begin{theorem}
		For any $\Lambda \geq \tfrac{1}{2}$, there is a \selfInclusiveGame~
		on a tree in which $\text{PoS}(G,r,1,\Lambda) = \tfrac{\maxdegree{G}}{2} = \tfrac{n-2}{2} = \tfrac{n-2}{b+1}$.
		\label{theorem:self_inclusive_pos_lower_bound}
	\end{theorem}
	\begin{proof}
		Consider the game $(G,r,1, \Lambda)$ with $\Lambda \geq \tfrac{1}{2}$ on a star-like graph~$G$ centered at~$v$ where one leaf node~$u$ is adjacent to one additional node~$w$. Hence, $\maxdegree{G} = n-2$. Moreover, assume that there is exactly one empty node, i.e., $e=1$. 
		
		There is a strategy profile $\sigma^*$ in which $\sigma^{-1}(v)$ is the blue agent and~$w$ is empty. We have $\text{DoI}(\sigma^*) = n$.
		However, we claim that the best NE has $\text{DoI}(\sigma) = 2$. The DoI can only be higher if the blue agent~$i$ is on a node with a degree of at least~2, i.e.,~$u$ or~$v$.
		If $\sigma(i) = u$, there must be two red agents adjacent to~$i$. Consequently, it holds that $U_i(\sigma) = \p{1}{3}$, yet the empty node~$u'$ must be adjacent to~$v$ and thus $U_i(\sigma_{iu'}) = \p{1}{2}$.
		If $\sigma(i) = v$, we have that either~$w$ is empty or there is a red agent~$j$ on~$w$ and some node~$u'$ adjacent to~$v$ is empty. In the first case, $U(\sigma_{iw}) = \p{1}{2} > U_i(\sigma) = \p{1}{r}$ and in the second one $U_j(\sigma_{ju'}) = \p{1}{2} > U_j(\sigma) = 0$. This proves that there can be no NE in which~$\sigma^{-1}(v)$ is blue.
		Therefore, $\text{PoS}(G,r,b,\Lambda) = \tfrac{\maxdegree{G}}{2}$.
	\end{proof}
	
	Next, we study the balanced game. Here, the PoS is upper bounded by a PoA of at most $2$. We show that this bound is tight for $\Lambda \geq \tfrac{1}{2}$.
	
	\begin{theorem}
		For $\Lambda \geq \tfrac{1}{2}$, a \selfInclusiveGame~$(G,b,b,\Lambda)$ with $\text{PoS}(G,b,b,\Lambda) \geq 2 - \varepsilon$ for any $\varepsilon > 0$ exists.
		\label{theorem:self_inclusive_pos_balanced_lower_bound}
	\end{theorem}
	\begin{proof}
		Consider the balanced game ($G,b,b,\Lambda)$ as shown in \Cref{fig:balanced_si_game_pos}. The graph~$G$ has two sets~$A$ and $B$ of~$b$ nodes each and the $i$-$th$ node in $A$ is connected to the $i$-th node in $B$. Furthermore, the first node~$v \in A$ is connected to all nodes in both, $A$ and $B$. Additionally, the node~$v$ is adjacent to~$2b$ leaf nodes $Z$.
		
		In the optimal strategy profile~$\sigma^*$, all red nodes are located on~$A$ and all blue nodes on~$B$. Thus, it holds that $\text{DoI}(\sigma^*) = 2b$.
		We claim that there is no equilibrium~$\sigma$ in which~$v$ is empty or any agent of the opposite color of~$\sigma^{-1}(v)$ is adjacent to any agent of~$\col{\sigma^{-1}(v)}$ other than~$\sigma^{-1}(v)$.
		
		Suppose that $v$ is empty. There are $2b-1$ nodes in $B \cup A \setminus \simpleset{v}$. Thus by counting, there must be an agent~$i$ on a node in~$Z$. As~$v$ is empty, we have that $U_i(\sigma) = 0$. Yet, it holds that $U_i(\sigma_{iv}) = \p{b}{2b} > 0$, so agent~$i$ has an improving jump.
		W.l.o.g., let there be a red agent on~$v$. Suppose that a blue agent~$i$ is adjacent to an additional red agent that is not~$\sigma(v)$. Then, we have that $U_i(\sigma) = \p{1}{3}$. As there are~$2b$ nodes in~$Z$ and neither~$\sigma^{-1}(v)$ nor~$i$ are on a node in~$Z$, there is an empty node~$u \in Z$ and since $U_i(\sigma_{iu}) = \p{1}{2} > U_i(\sigma)$ agent~$i$ has an improving jump.
		Thus, $v$ is not empty and all red agents except for~$\sigma^{-1}(v)$ are segregated. For any equilibrium, it holds that at most the agent on~$v$ and the agents of a different color may be non-segregated, i.e., $\text{DoI}(\sigma) \leq b+1$. In \Cref{fig:balanced_si_game_pos} we present such a NE: All nodes in $A$ are occupied by red agents, all nodes in $B$ are empty and~$b$ nodes in $Z$ are blue. Clearly, no red agent can improve and any blue agent jumping to a node in $B$ will have a utility of either~$\p{1}{2}$ or~$\p{1}{3}$ which is not better than her current utility.
		Hence, $\text{PoS}(G,b,b,\Lambda) \geq \tfrac{2b}{b+1}$. Thus, for any~$\varepsilon > 0$, we can achieve a $PoS \geq 2-\varepsilon$ by choosing~$b$ large enough.
	\end{proof}
	\begin{figure}[t]
		\centering
		\includegraphics[width=0.6\textwidth]{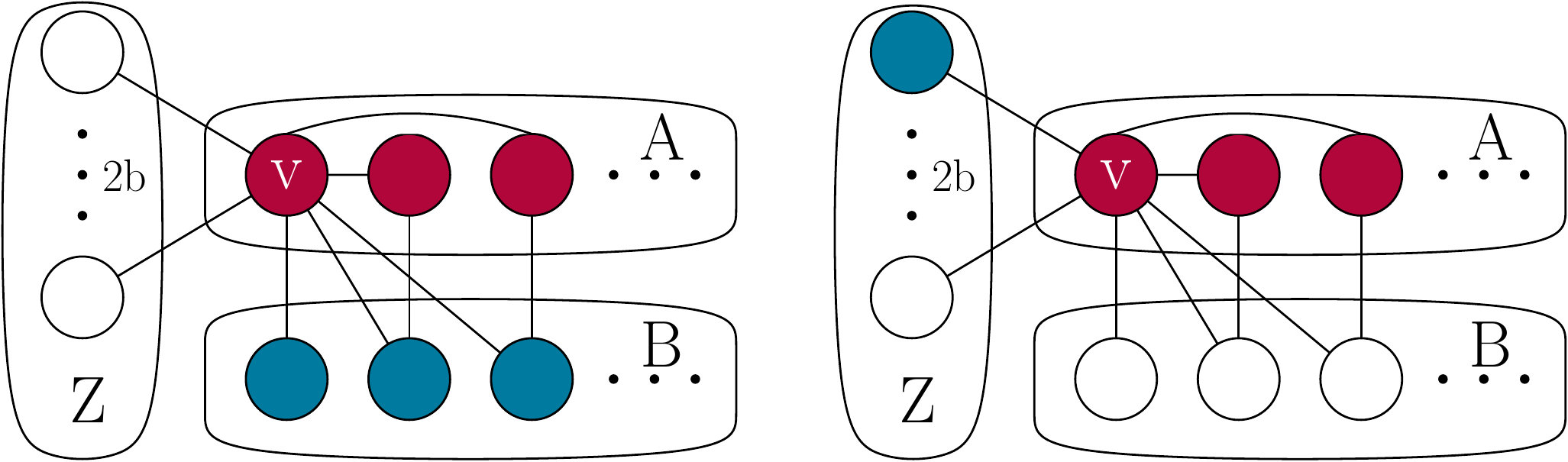}
		\caption{The PoS of a balanced \selfInclusiveGame~with $b$ agents per type. $A$ and $B$ contain $b$ nodes each. Left: Optimum. Right: Best NE, where $b$ nodes in $Z$ are occupied by blue agents.}
		\label{fig:balanced_si_game_pos}
	\end{figure}
	
	Earlier, in \Cref{theorem:independent_set}, we proved the existence of equilibria for graphs that have an independent set of size at least $b+e$. Now, we show that on such graphs, we can also bound the $PoS$.
	
	\begin{theorem}
		For any \selfInclusiveGame~$(G,r,b,\Lambda)$ with $b + e \leq \alpha(G)$, we have $\text{PoS}(G,r,b,\Lambda) \leq \tfrac{\maxdegree{G}+1}{\mindegree{G}+1}.$
		\label{theorem:self_inclusive_pos_independent_set_upper_bound}
	\end{theorem}
	\begin{proof}
		In \Cref{theorem:independent_set} we describe the construction of NE in which all blue agents are adjacent to only red agents. Therefore, it holds for the best NE~$\sigma$ that $\text{DoI}(\sigma) \geq (\mindegree{G}+1) b$. Furthermore, by \Cref{lemma:self_inclusive_poa_doi_upper_bound}, we have that for the optimal strategy profile~$\sigma^*$, it holds $\text{DoI}(\sigma^*) \leq b(\maxdegree{G}+1)$. Thus, 
		$\text{PoS}(G,r,b,\Lambda) \leq \tfrac{b(\maxdegree{G}+1)}{b(\mindegree{G}+1)} = \tfrac{\maxdegree{G}+1}{\mindegree{G}+1}.$
	\end{proof}
	In particular, this applies to $\delta$-regular graphs since $\maxdegree{G}= \mindegree{G} = \delta$. Note that for any~$\delta$-regular graph, we have $\alpha(G) \geq \tfrac{\abs{V}}{\delta+1} $.
	\begin{corollary}
		For any \selfInclusiveGame~on a $\delta$-regular graph with $b+e \leq \alpha(G)$, we have $\text{PoS}(G, r,b, \Lambda) = 1$.
		\label{corollary:self_inclusive_pos_independent_set_regular_upper_bound}
	\end{corollary}
	
	Furthermore, in \Cref{theorem:rings_self_inclusive_e1_convergence}, we prove that any game on a ring with $\Lambda = \tfrac{1}{2}$ and $e=1$ converges to a NE by proving that $\text{DoI}(\sigma)$ is an ordinal potential function. It follows that every strategy profile that maximizes the degree of integration must be a NE.
	\begin{corollary}
		For any \selfInclusiveGame~$(G,r,b, \tfrac{1}{2})$ on a ring with $e=1$, we have $\text{PoS}(G,r,b,\tfrac{1}{2}) = 1$.
	\end{corollary}
	
	\subsection{Quality of Equilibria with Respect to the Utilitarian Welfare}
	
	While our main focus in this work is on the quality of equilibria with respect to the degree of integration as social welfare, we close this section by pointing out, that our results on the PoA and PoS with respect to the degree of integration also imply bounds on the PoA and the PoS with respect to the standard utilitarian welfare ($\text{PoA}^{U}$ and $\text{PoS}^{U}$ for short), assuming that $p$ is linear. Remember, that the utilitarian social welfare simply is the sum over the utilities of all the agents. 
	
	In particular, for a fixed peak $\Lambda$ and a fixed maximum degree $\delta$, a constant bound on PoA yields a constant bound on $\text{PoA}^{U}$, as the following theorem demonstrates. 
	
	\begin{restatable}{theorem}{theoremsixone}\label{theorem:self_inclusive_poa_doi_implies_poa_sum_of_utilities} 
		Let $p$ be a linear function. For any \selfInclusiveGame\ $\Gamma=(G,r,b,\Lambda)$, the following holds:
		\begin{itemize}
			\item $\text{PoA}(\Gamma) \leq a \Rightarrow PoA^{U}(\Gamma) \leq a \cdot \max(\Lambda, (1-\Lambda)) \cdot (\maxdegree{G}+1).$
			\item $\text{PoS}(\Gamma) \leq s \Rightarrow PoS^{U}(\Gamma) \leq s \cdot \max(\Lambda, (1-\Lambda)) \cdot (\maxdegree{G}+1).$
		\end{itemize}
		For the PoA, this bound is asymptotically tight, i.e., $PoA^{U}(G,b,b,\tfrac{1}{2}) = \text{PoA}(G,b,b,\tfrac{1}{2}) \cdot \tfrac{1}{2} \cdot \maxdegree{G}$ holds.
	\end{restatable}
	\begin{proof}
		Any non-segregated agent has a utility larger than zero. The lowest possible positive utility is bounded by the maximum degree of the graph and is $\p{1}{\maxdegree{G}+1} = \tfrac{1}{\Lambda\cdot (\maxdegree{G}+1)}$ for $\Lambda \geq \tfrac{1}{2}$, respectively $\p{\maxdegree{G}}{\maxdegree{G}+1} = \tfrac{1}{(1-\Lambda)\cdot (\maxdegree{G}+1)}$ for $\Lambda \leq \tfrac{1}{2}.$ Hence, the ratio between the worst possible utility and the highest possible utility of a non-segregated agent is $\max(\Lambda, (1-\Lambda)) \cdot (\maxdegree{G}+1)$. For the sake of readability, let $\max(\Lambda, (1-\Lambda)) = m_{\Lambda}$.
		
		Let~$\sigma$ be the worst NE with respect to the sum of utilities and~$\sigma'$ be the worst NE with respect to the DoI. Thus, it holds that $\text{DoI}(\sigma) \geq \text{DoI}(\sigma')$. Hence, it follows that
		$$\sum_i U_i(\sigma) \geq \tfrac{\text{DoI}(\sigma)}{m_{\Lambda} \cdot (\maxdegree{G}+1)} \geq \tfrac{\text{DoI}(\sigma')}{m_{\Lambda} \cdot (\maxdegree{G}+1)}.$$
		Let~$\sigma^*$ be the best strategy profile with respect to the sum of utilities and~$\sigma^{*'}$ be the best strategy profile with respect to the DoI. This means that $\text{DoI}(\sigma^{*'}) \geq \text{DoI}(\sigma^*)$ and therefore $$\sum_i U_i(\sigma^*) \leq \text{DoI}(\sigma^*) \leq  \text{DoI}(\sigma^{*'}).$$
		
		It holds that
		\begin{align*}
			\text{PoA}^U(\Gamma) &= \tfrac{\sum_i U_i(\sigma^*)}{\sum_i U_i(\sigma)} 
			\leq  \tfrac{\text{DoI}(\sigma^{*\prime}) m_{\Lambda} \cdot (\maxdegree{G}+1)}{\text{DoI}(\sigma')} 
			= \text{PoA}(\Gamma) \cdot m_{\Lambda} \cdot (\maxdegree{G}+1).
		\end{align*}
		Let~$\sigma$ be the best NE with respect to the sum of utilities and~$\sigma'$ be the best NE with respect to the DoI. It holds that $$\sum_i U_i(\sigma) \geq \sum_i U_i(\sigma') \geq \tfrac{\text{DoI}(\sigma')}{m_{\Lambda} \cdot (\maxdegree{G}+1)}.$$
		This also applies to the PoS.
		\begin{align*}
			\text{PoS}^U(\Gamma) &= \tfrac{\sum_i U_i(\sigma^*)}{\sum_i U_i(\sigma)} 
			\leq  \tfrac{\text{DoI}(\sigma^{*\prime}) m_{\Lambda} \cdot (\maxdegree{G}+1)}{\text{DoI}(\sigma')} 
			\leq \text{PoS}(\Gamma) \cdot m_{\Lambda} \cdot (\maxdegree{G}+1).
		\end{align*}
		It remains to show that the bound for the PoA is asymptotically tight.
		For this, consider the balanced game $(G,b,b,\tfrac{1}{2})$ as illustrated in \Cref{fig:poa_sum_of_utilities}.
		\begin{figure}[h]
		\centering
		\includegraphics[width = 0.8\textwidth]{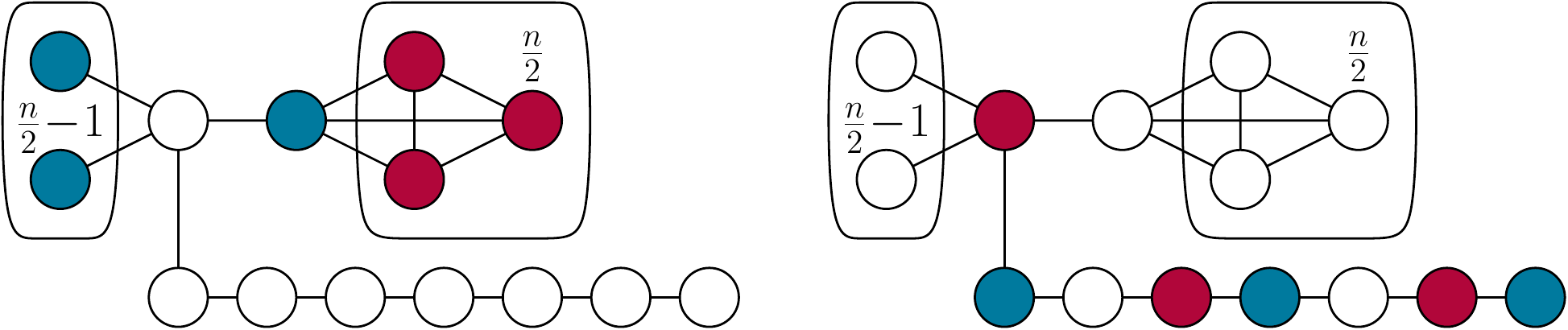}
		\caption{A \selfInclusiveGame~$(G,b,b,\tfrac{1}{2})$ for which the $\text{PoA}^U(G,b,b,\tfrac{1}{2})$ with respect to the sum of utilities (assuming that~$p$ is a linear function) is $\tfrac{n(\maxdegree{G}+1)}{2(b+1)}$. Note that $\maxdegree{G}$ is~$b+1$. Left: A NE in which all red agents have utility $ \p{b}{b+1} =  \p{1}{\maxdegree{G}+1}$. One blue agent also has this utility while all other blue agents have utility zero.
			Right: A social optimum with respect to the sum of utilities. All agents have the highest possible utility $\p{1}{2}=1$.}
		\label{fig:poa_sum_of_utilities}
	\end{figure}
	The graph~$G$ has a path of length~$3b$, on which we can place the agents in pairs of two, therefore, in the optimal strategy profile $\sigma^*$, it holds for all agents $i$ that $U_i(\sigma^*) = 1$, i.e., $\text{DoI}(\sigma^*) = \sum_i U_i(\sigma^*) = n$. However, there is a NE~$\sigma$ in which the sum of utilities matches the bounds we derived earlier. The graph~$G$ has a clique with~$b+1$ nodes, on which we place all red agents and one blue agent. The node of the blue agent is connected to an empty node~$v$, which is adjacent to~$b-1$ blue agents on nodes of degree one and the path. Note that therefore~$v$ is the only empty node adjacent to any agent. Hence, no blue agent has an improving jump. Furthermore, each red agent~$i$ has a utility of $U_i(\sigma) = \p{b}{b+1} = \p{1}{\maxdegree{G}} = U_i(\sigma_{iv})$ and therefore no improving jump. Thus, we have that~$\sigma$ is a NE. As the blue agent in the clique has the same utility as the red agents, we have $\sum U_i(\sigma) = (b+1) \cdot \p{1}{\maxdegree{G}} = \p{2(b+1)}{\maxdegree{G}}$. Hence, it follows that $\text{PoA}^U(G,b,b,\tfrac{1}{2}) = \tfrac{n \cdot \maxdegree{G}}{2(b+1)}$ while it holds that $\text{PoA}(G,b,b,\tfrac{1}{2}) = \tfrac{n}{b+1}$, giving us that $$\text{PoA}^U(G,b,b,\tfrac{1}{2}) = \text{PoA}(G,b,b,\tfrac{1}{2}) \cdot m_{\Lambda} \cdot \maxdegree{G}. \qedhere$$
	\end{proof}

	\section{Computational Complexity}
	In this section we discuss the computational complexity of finding equilibria via improving response dynamics and the complexity of computing strategy profiles with a high DoI. As already pointed out in Our Contribution, we believe that especially the former question is particularly interesting, since finding equilibria via improving moves can be easily coordinated within a society of selfish agents. In contrast, centrally switching from some initial state directly to an equilibrium state requires much more coordination and also that the agents trust the central coordinator. 
	
	Settling the complexity of the equilibrium decision problem seems to be very challenging and we leave this as an open problem. However, our hardness proof for finding equilibria via improving response dynamics can be seen as a first step towards proving that deciding the existence of equilibria is NP-hard as well. Moreover, we note in passing that if we would allow for stubborn agents, as in~\citep{A+19}, then we can prove that deciding if an equilibrium exists is indeed NP-hard. We suspect that this assumption may be removed, similarly to the approach of \cite{KBFN22}.
	
	\subsection{{Finding Equilibria via Improving Response Dynamics}}
	
	We start with investigating the problem of finding equilibria. To this end, we consider the problem of deciding whether an equilibrium for a given game can be reached through \emph{improving response dynamics (IRDs)} from a given initial strategy profile~$\sigma_0$.
	We show that this problem is NP-hard for any value of $\Lambda$. For the sake of presentation, we start with proving hardness for the case $\Lambda = \frac12$.

	\subsubsection{Hardndess for $\Lambda = \frac12$}

	We show the hardness of finding equilibria through IRDs by a reduction from the NP-complete problem \textsc{Double 4-SAT}.
	\begin{definition}[\textsc{Double 4-SAT}]
		Given a SAT formula in which each clause consists of~4 literals, decide if there is an assignment in which at least~2 literals in each clause are true.
	\end{definition}
	\textsc{Double 4-SAT} is NP-complete by a reduction from \textsc{3-SAT} (\cite{Karp1972}). 
	Let~$\varphi$ be an instance of \textsc{Double 4-SAT} with variables $x_1, \dots, x_k$ ($k \geq 3$) and clauses $c_1, \dots, c_m$.
	We define $\Gamma_{\varphi} = (G, r, b, \tfrac{1}{2})$ as a corresponding game and $\sigma^{\varphi}_{0}$ as its initial strategy profile.  
	
		We first provide a detailed description of the construction used for the reduction. Consider \Cref{fig:np_hardness_1_2_lambda} for an illustration.
	\begin{definition}
		Let $\varphi$ be an instance of \textsc{Double 4-SAT} with variables $x_1, \dots x_k$ ($k \geq 3$) and clauses $c_1, \dots c_m$.
		We define $\Gamma_{\varphi} = (G, r, b, \tfrac{1}{2})$ as a corresponding game and $\sigma^{\varphi}_{0}$ as its initial strategy profile. The graph $G$ is constructed in the following way.
		\begin{itemize}
			\item Let $z = 11\cdot 2 k + 10m + 3mk$. The graph has a clique $Z$ with $2z$ nodes, split into two disjoint subsets $Z_R, Z_B$ of size $z$ each.
			\item For every variable $x_i$, there is a pair of adjacent nodes $x_i$ and $\overline{x_i}$. Let $X$ be the set of these nodes. Each of it is adjacent to 5 nodes in $Z_B$ and 11 nodes in $Z_R$.
			\item There is a clique $C$ of $m$ nodes, each node $c_i \in C$ corresponds to one clause $c_i$ in $\varphi$. Each of these nodes is adjacent to 5 nodes in $Z_B$ and 10 nodes in $Z_R$. Furthermore, a node $c_i$ is adjacent to the nodes in $X$ corresponding to the literals in the clause $c_i$.
			\item For each node $c_i \in C$, there is a group $Y_i$ of $k$ nodes adjacent to $c_i$. Let $Y$ be the set of all these nodes and let each node $y \in Y$ be adjacent to 3 nodes in $Z_B$.
			\item Each node in $Z$ is adjacent to at most one node outside $Z$. We have chosen $z$ sufficiently high.
		\end{itemize}
		Let $r = z+k$ and $b = z$. For the initial placement $\sigma^{\varphi}_{0}$, we have that $\sigma_0^{-1}(v)$ is a red agent if and only if $v \in Z_R \cup Y_0$ and a blue agent if and only if $v \in Z_B$.
	\end{definition}
	
	\begin{figure}[t]
		\centering
		\includegraphics[width = 0.95\textwidth]{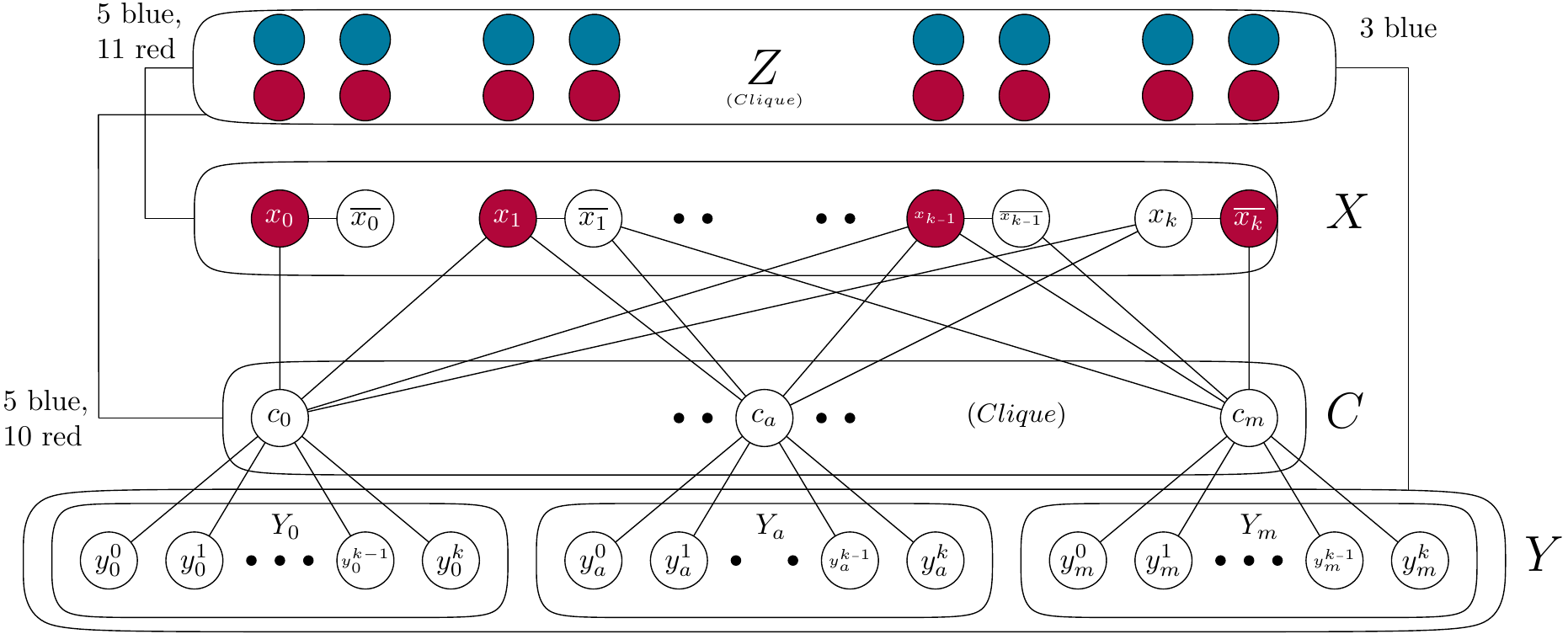}
		\caption{The \selfInclusiveGame~$\Gamma(\varphi)$ corresponding to an instance~$\varphi$ of the \textsc{Double 4-SAT} Problem with~$k$ variables and~$m$ clauses.
			The nodes in the set~$X$ correspond to the literals and the nodes in~$C$ to the clauses. $Z$ is a clique of $z$ red and $z$ blue agents with $z = 22 k + 10m + 3mk$.
			Each node in $X$ is connected to $11$ nodes of $Z_R$ and $5$ nodes of $Z_B$. Each node in the clique $C$ is connected to $10$ of $Z_R$ and $5$ of $Z_B$. Each set of nodes $Y_a$ contains $k$ nodes each that are connected to $c_a \in C$ and $3$ nodes in $Z_B$. Each node in $Z$ is only adjacent to at most one node outside of $Z$.
			In the initial strategy profile $\sigma^\varphi_0$, all red agents start on $Z_R$ and $Y_0$ and all blue agents on $Z_B$.}
		\label{fig:np_hardness_1_2_lambda}
	\end{figure}

	We start with a few observations that hold for any strategy profile $\sigma$ for which all nodes in $Z_R$ are occupied by red agents and all nodes in $Z_B$ are occupied by blue agents. In particular, this implies that any agent outside of $Z$ is red and any agent $i$ located on a node in $X$ or $C$ is adjacent to more red than blue agents. Hence, $f_i(\sigma) > \tfrac{1}{2}$. Thus, the more adjacent red agents outside $Z$ an agent occupying a node in $X$ and $C$, respectively, has, the lower is her utility.

	 Hence, under the assumption that all nodes in $Z_R$ are occupied by red agents and all nodes in $Z_B$ are occupied by blue agents, it holds for an agent $i$ that
	 \begin{enumerate}
	 	\item if $\sigma(i) \in X$, $i$ has a utility of at most $U^{X}_{\max} = \p{12}{17}$,
	 	\item if $\sigma(i) \in C$, $i$ has a utility of at most $U^{C}_{\max} = \p{11}{16}$, and
	 	\item  if $\sigma(i) \in Y_a$, $i$ has a utility of $U^{Y}_{\max} = \p{2}{5}$ if $c_a$ is occupied and $U^{C}_{\min} =\p{1}{4}$ otherwise.
	 \end{enumerate}
	Consequently, we have $1 > U^{Y}_{\max} > U^{C}_{\max} > U^{X}_{\max} > U^{Y}_{\min}$. Note, that $U^{X}_{\max}$ is the second highest utility obtainable on nodes in~$C$.
	We can show that agents starting on nodes in $Z$ have a higher utility than they could achieve by jumping to a node outside of $Z$.
	
	\begin{restatable}{lemma}{lemmafivethree}
		Let $\sigma$ be a strategy profile that is identical to the initial placement $\sigma^{\varphi}_{0}$ on all nodes in $Z$. No agent on a node in $Z$ has an improving jump.
		\label{lemma:hardness_1_2_no_improving_jump_z_NE}
	\end{restatable}
	\begin{proof}
		Assume for the sake of contradiction that $i$ is an agent with $\sigma(i) \in Z$ that has an improving jump to node $v$. Because $\sigma$ is identical to $\sigma^{\varphi}_{0}$ with respect to $Z$, $v$ must be in $X \cup C \cup Y$ and all nodes in $ X \cup C \cup Y$ must either empty or occupied by a red agent. Furthermore $\sigma(i)$ can be adjacent to at most one node $u \not\in Z$.
		We begin with the observation that $v$ cannot be $u$, since if $u$ is empty, we have that $U_i(\sigma) = \p{1}{2}=1$, contradicting the existence of an improving jump.
		Consequently, node~$v$ cannot be adjacent to $\sigma(i)$ and node~$u$ is occupied by a red agent.
		
		If $i$ is red, we have $U_i(\sigma) = \p{z+1}{2z+1}$, yet it holds that $U_i(\sigma_{iv}) \leq U^{Y}_{\max} = \p{2}{5}$. Since $z > 2$, we have that $|\tfrac{1}{2} - \tfrac{2}{5}| > |\tfrac{z+1}{2z+1} - \tfrac{1}{2}|$ and thus $U_i(\sigma_{iv}) < U_i(\sigma)$.
		
		If $i$ is blue, $U_i(\sigma) = \p{z}{2z+1}$. The highest utility for the blue agent $i$ on a node in $X$ is $\p{6}{16}$, the highest utility on a node in $C$ is $\p{6}{17}$ and the highest utility on a node in $Y$ is $\p{4}{5}$, so overall $U_i(\sigma_{iv}) \leq \p{6}{16}$. Since $z \geq 2$, it holds that $|\tfrac{1}{2} - \tfrac{6}{16}| > |\tfrac{1}{2} - \tfrac{z}{2z+1}|$ and therefore $U_i(\sigma_{iv}) < U_i(\sigma)$.
	\end{proof}
	Hence, all agents placed on $Z$ behave like stubborn agents, i.e., they do not jump.
	\begin{corollary}
		Starting from $\sigma^{\varphi}_{0}$, every NE reached through improving response dynamics must be identical to~$\sigma^{\varphi}_{0}$ on all nodes in~$Z$.
		\label{corollary:hardness_1_2_stubborn_agent_replacement}
	\end{corollary}

	The next lemma provides necessary conditions for any NE.
	\begin{restatable}{lemma}{lemmafivefive}
		Let $\sigma$ be a strategy profile for $\Gamma_{\varphi}$ that is identical to~$\sigma^{\varphi}_{0}$ on all nodes in~$Z$.
		Then $\sigma$ cannot be a NE, if
		(1) there is an agent $i$ with $\sigma(i) \in C$, or
		(2) there are agents $i,j$ with $\exists l: \sigma(i) = x_l, \sigma(j) = \overline{x_l}$, or
		(3) there is an agent $i$ with $\sigma(i) \in Y$,
		\label{lemma:hardness_1_2_no_equilibrium_conditions}
	\end{restatable}
	\begin{proof}
		We prove that under the conditions (1) -- (3), $\sigma$ cannot be a NE.
		\begin{enumerate}
			\item Assume that there is an agent $i$ with $\sigma(i) = c_a \in C$. \\
			Consider the case that it holds for all agents $j$ with $j \neq i$ that $\sigma(j) \notin C$. Then, all agents which are not located on nodes in $Z$ must be on $Y_a$ as otherwise jumping to a node in $Y_a$ improves their utility. Consequently, agent~$i$ has a utility of $\p{11+k-1}{16+k-1}$. As $k \geq 3$, this is lower than $U^{X}_{\max}=\p{12}{17}$. Hence,~$i$ has an incentive to jump to an arbitrary node in $X$.\\
			Consider now the case where there is an agent $j \neq i$ with $\sigma(j) \in C$. Since $\abs{Y_a} = k$, there must be an empty node~$y$ in~$Y_a$. We have $U_j(\sigma) \leq U^{C}_{\max} < U^{Y}_{\max} = U_i(\sigma_{jy})$, so $\sigma$ cannot be a NE.
			
			\item Assume that there are agents $i, j$ with $\exists l: \sigma(i) = x_l, \sigma(j) = \overline{x_l}$. If any node in $C$ is occupied by an agent, condition (1) shows that it cannot be a NE. Otherwise, both $i$ and $j$ have a utility of $\p{13}{18}$. Yet, by counting there must be a pair of empty nodes $x_p, \overline{x_p} \in X$ and $U_i(\sigma_{ix_p}) = \p{12}{17}  > U_i(\sigma)$.
			
			\item Assume there is an agent $i$ with $\sigma(i) \in Y$. If there is an agent on a node in $C$, it cannot be a NE due to condition (1). Therefore, $U_i(\sigma) = \p{1}{4}$. Also by counting we get that there must be a pair of empty nodes $x_p, \overline{x_p} \in X$, so $U_i(\sigma_{ix_p})= \p{12}{17} > U_i(\sigma)$.
		\end{enumerate}
	\end{proof}
	
	We now provide our hardness result for $\Lambda = \tfrac{1}{2}$ for finding NE via IRDs.
	
	\begin{theorem}
		It is NP-hard to decide if a given game played on a graph $G$ with $r$ red and $b$ agents and peak $\Lambda = \tfrac{1}{2}$ can reach a NE through IRDs starting from an initial placement $\sigma^{\varphi}_{0}$.
		\label{theorem:np_hardness_lambda_1_2}
	\end{theorem}
	\begin{proof}
		Let $\varphi$ be a satisfiable instance of \textsc{Double 4-SAT} with~$k$ variables. Consider the game $\Gamma_\varphi$ and strategy profile $\sigma$ identical to $\sigma^{\varphi}_{0}$ with respect to $Z$, in which the other $k$ red agents are placed on the nodes corresponding to the true literals of a satisfying assignment for $\varphi$.
		We want to show that $\sigma$ is a NE. It follows from \Cref{lemma:hardness_1_2_no_improving_jump_z_NE} that no agent $i$ with $\sigma(i) \in Z$ has an improving jump. It remains to show that no agent $i$ with $\sigma(i) \in X$ has an improving jump to a node in $X, C$ or $Y$.
		
		Let $i$ be an agent with $\sigma(i) \in X$. In $\sigma$, all agents outside $Z$ are on non-adjacent nodes in $X$ and hence, $U_i(\sigma) = U^{X}_{\max}$. Thus, no jump to a node in $X$ can be improving for $i$.
		In a satisfying assignment, at least two literals per clause are true. Therefore, we have that every $c_a \in C$, is adjacent to at least two nodes in $X$ that are occupied by a red agent. Let $c_a \in C$. Thus, there is an agent $j \neq i$ with $\sigma(j)$ adjacent to $c_a$ and consequently it holds that $U_i(\sigma_{iv}) \leq \p{12}{17} = U^{X}_{\max} = U_i(\sigma)$. This means that $i$ has no improving jump to a node in $C$.
		Furthermore, as all nodes in $C$ are empty, all nodes in $Y$ offer at most a utility of $U^Y_{\min} < U_i(\sigma)$.
		Hence, $\sigma$ is a NE. It can be reached with IRDs, as all $k$ red agents outside of $Z$ start on $Y$ with a utility of $U^Y_{\min}$ and can, one after another, perform an improving jump to the appropriate position on~$X$.
		
		Let $\sigma$ be a NE for $\Gamma_\varphi$, reached through IRDs starting from $\sigma^{\varphi}_{0}$. According to \Cref{corollary:hardness_1_2_stubborn_agent_replacement} and \Cref{lemma:hardness_1_2_no_equilibrium_conditions}, no agent is placed on a node in $Y$ or $C$. Furthermore, for each of the $k$ variables, exactly one literal node is occupied by one of the $k$ strategic agents (as no two literal nodes belonging to the same variable can both be occupied) and all of these $k$ agents have a utility of $U^X_{\max}$.
		Assume for the sake of contradiction that one clause node $c_i \in C$ is not adjacent to at least two red agents in $X$. If it is adjacent to no such agent, all agents have an incentive to jump there, as $U^C_{\max} > U^X_{\max}$. If it is adjacent to exactly one agent $j$ in $X$, agent~$j$ can jump to $c_i$ and will have a utility of $U^C_{\max}$ afterward. Thus all clause nodes are adjacent to at least two red agents and thus $\varphi$ is a satisfiable Double 4-SAT instance.
	\end{proof}

	\subsubsection{Hardness for arbitrary $\Lambda \in (0,1)$}
	In the following, we prove that it is NP-hard to find equilibria through improving response dynamics, even for an arbitrary fixed value of $\Lambda$. The general idea of the proof is the same as for proving the hardness for $\Lambda = \tfrac{1}{2}$; yet, the construction gets more complicated. This is the case since for $\Lambda= \tfrac{1}{2}$, the gadget $Y$ has nodes which either have a utility equals $\p{1}{4}$ or $\p{2}{5}$. However, for an arbitrary value of $\Lambda$ this gap gets much smaller, forcing us to use a larger number of nodes in $Z$ to ensure that $U^X_{\max}$ and $U^C_{\max}$ are in between these two values. For a better overview, we provide a rough sketch of the utilities in \Cref{fig:np_hardness_arbitrary_lambda_utilities}.
	\begin{figure}[h]
		\centering
		\includegraphics[width = 0.9\textwidth]{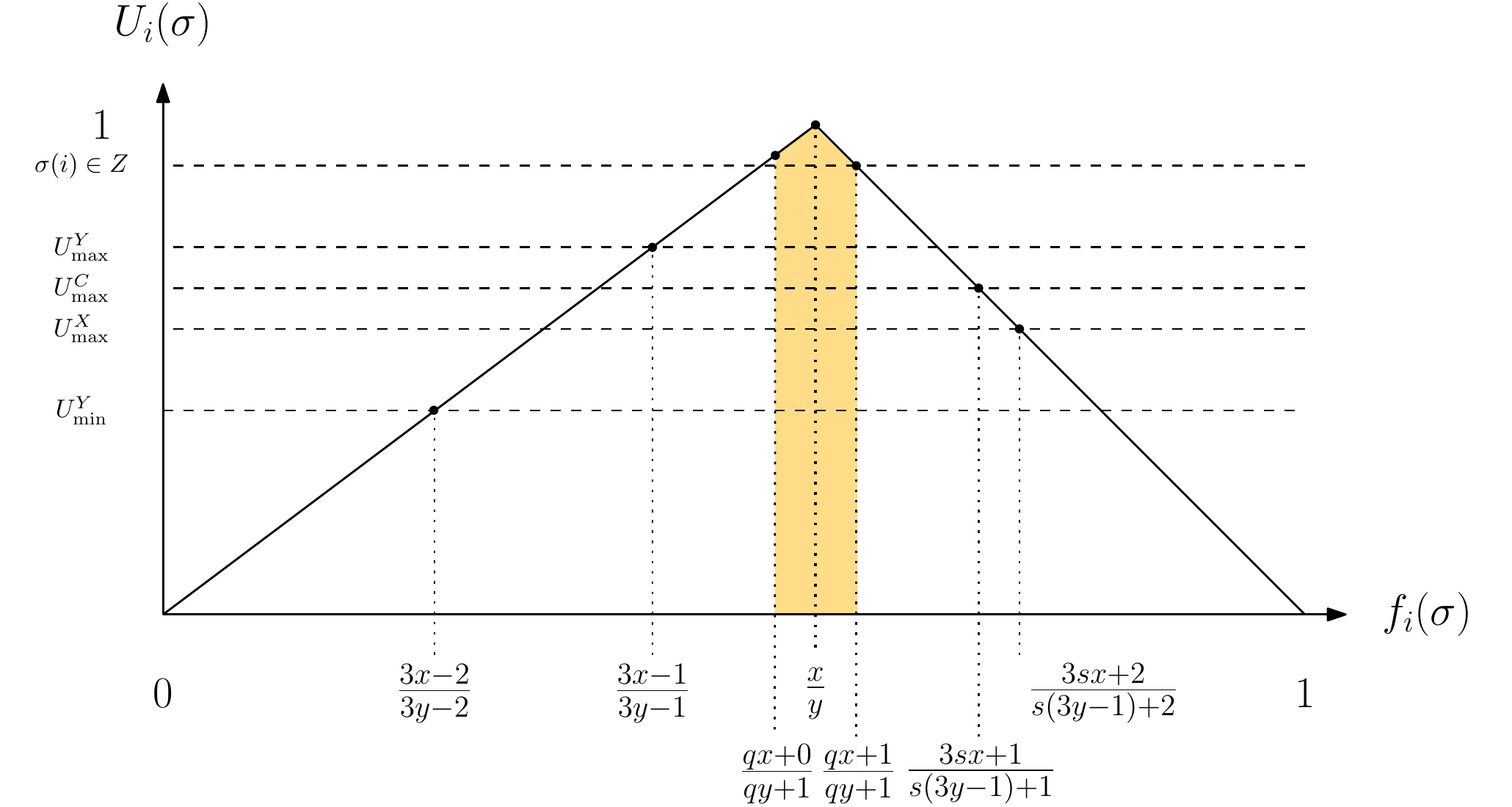}
		\caption{Illustration of the utilities of agents in a strategy profile $\sigma$ for $\Gamma(\varphi,\tfrac{x}{y},q)$, assuming that all nodes in $Z_R$ are occupied by red and all nodes in $Z_B$ are occupied by blue agents.
			Yellow: Utility of an agent on $Z$ (either $\p{x}{y}, \p{qx+1}{qy+1}$ or $\p{qx}{qy+1}$. For a sufficiently large value of $q$, this is arbitrarily close to 1.
			A red agent on a node in $C$ has a utility of at most $U_{\max}^C =\p{s3x+1}{s(3y-1)+1}$. The maximum utility of a red agent on a node in $X$ is slightly lower with $U_{\max}^X = \p{s3x+2}{s(3y-1)+2}$.
			A red agent on a node in $Y_a$ has two possible utilities, $U_{\min}^Y =\p{3x-2}{3y-2} = \p{3x}{3y-2}$ if $c_a$ is empty and $U_{\max}^Y = \p{3x-1}{3y-1} = \p{3x}{3y-1}$ if there is a red agent on $c_a$. Observe that $U_{\max}^X, U_{\max}^Y$ are in between these values.}
		\label{fig:np_hardness_arbitrary_lambda_utilities}
	\end{figure}
	\begin{figure}[h]
		\centering
		\includegraphics[width = 0.5\textwidth]{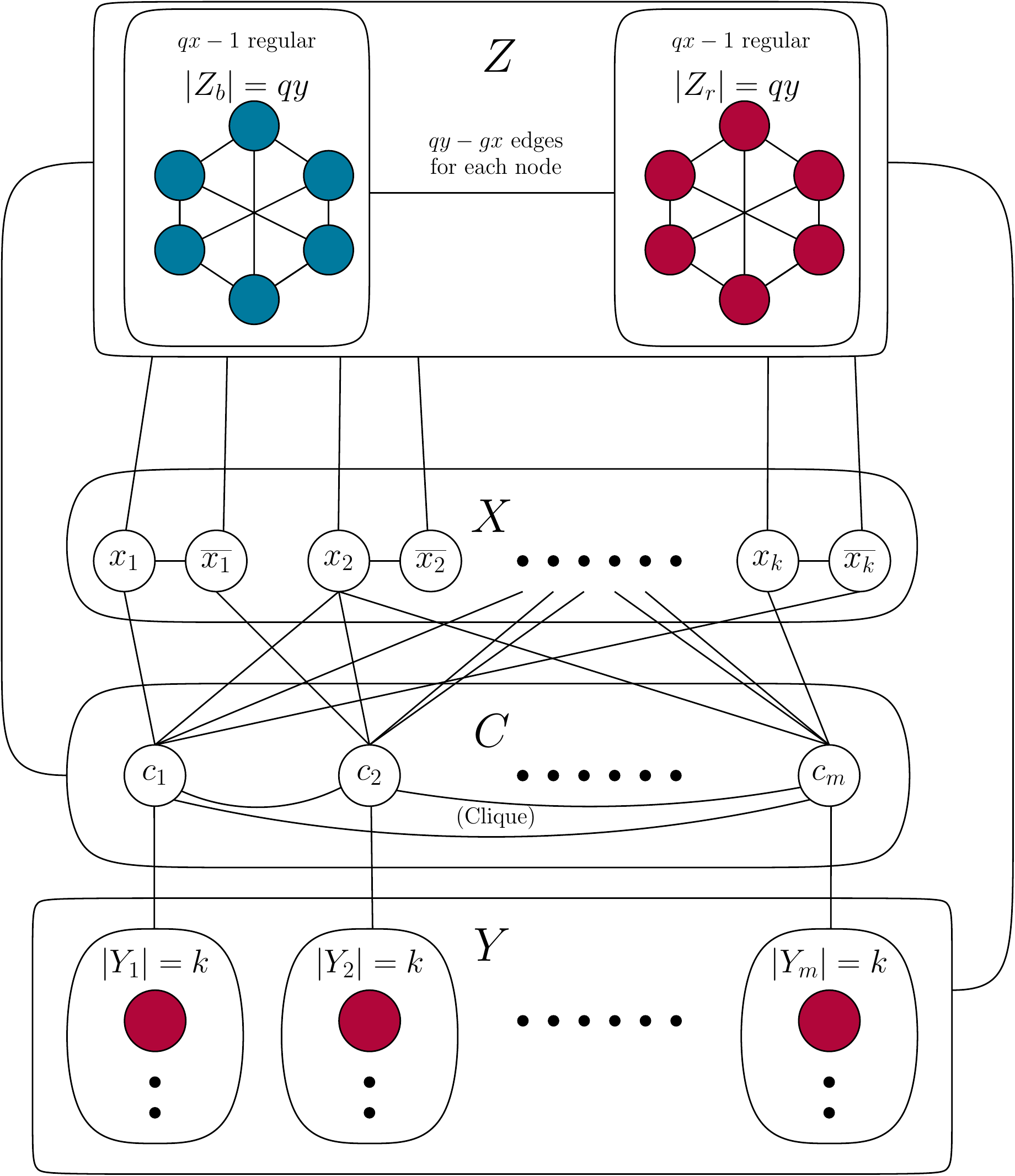}
		\caption{The \selfInclusiveGame~$\Gamma(\varphi,\Lambda,q)$ and initial strategy profile $\sigma^{\Gamma(\varphi, \Lambda,q)}$ corresponding to an instance~$\varphi$ of the \textbf{Double 4 SAT Problem} with~$k$ variables and~$m$ clauses.
			The nodes in the set~$X$ correspond to the literals and the nodes in $C$ to the clauses. The agents on the nodes in $Z$ are arranged such that they will never have improving jumps.
			Let $s = 6yk$ and let $q$ be sufficiently large.
			Each node in $X$ is connected to $3sx+1$ nodes of $Z_R$ and $s(3y-1) - 3sx$ nodes of $Z_B$. Each node in the clique~$C$ is connected to $3sx$ nodes of $Z_R$ as well as $s(3y-1) - 3sx$ nodes of $Z_B$. Each set of nodes $Y_a$ contains $k$ nodes each that are connected to $c_a \in C$, $3x-3$ nodes in $Z_R$ and $3(y-x)$ nodes in $Z_B$. Each node in $Z$ is only adjacent to one node outside $Z$.
			Nodes marked in red/blue are occupied by a red/blue agent in the initial strategy profile $\sigma^{\Gamma(\varphi, \Lambda,q)}$.
			\Cref{theorem:hardness_arbitrary_lambda} proves that $\varphi$ has a satisfying assignment if and only if a NE can be reached through improving response dynamics starting from $\sigma^{\Gamma(\varphi, \Lambda,q)}$.
		}
		\label{fig:np_hardness_arbitrary_lambda}
	\end{figure}

	We start by defining our generalized mapping of \textsc{Double 4 SAT} instances to the \selfInclusiveGame:
	\begin{definition}
		Let $\varphi$ be an instance of \textsc{Double 4 SAT} with variables $x_1, \dots x_k$ ($k > 2$) and clauses $c_1, \dots c_m$.
		For a fixed rational number $\Lambda \in (0,1)$, let $x,y \in \mathbf{N}$, such that $\tfrac{x}{y} = \Lambda$. W.l.o.g., we can assume that $x \neq \tfrac{y(3y+1)}{6y-2}$ and $x \neq \tfrac{y(3y+1)}{6y-1}$.
		\footnote{Assume that we have chosen $x,y$ s.t. $x = \tfrac{y(3y+1)}{6y-2}$ (or $x = \tfrac{y(3y+1)}{6y-1}$). Then, for any $d > 1, x'=dx, y'=dy$, it holds that $dx = d \tfrac{y(3y+1)}{6y-2} = x' = \tfrac{y'(3y'+1)}{6y'-2}$ if and only if $y = 0$. However, we have $y>x> 0$. Hence, by choosing another representation $\tfrac{x'}{y'} = \Lambda$ we can fulfill this requirement.}

		Further, let $q \in \mathbf{N}$ with $qy(qx-1) \mod 2 = 0$.

		We define $\Gamma(\varphi,\Lambda,q) = (G, r, b, \Lambda)$ as a corresponding \selfInclusiveGame~and $\sigma^{\Gamma(\varphi, \Lambda,q)}$ as the initial strategy profile of this \selfInclusiveGame~in the following way, cf. \Cref{fig:np_hardness_arbitrary_lambda}.

		\begin{itemize}
			\item $G$ is a graph with two sets of nodes $Z_R$ and $Z_B$, $|Z_B| = |Z_R| = qy$, $Z = Z_R \cup Z_B$. Both induce a $qx-1$ regular graph. Note that this is possible due to $qy(qx-1) \mod 2 = 0$. Furthermore, each node in $Z_R$ (resp. $Z_B$) is adjacent to exactly $qy-qx$ nodes of $Z_B$ (resp.~$Z_R$).
			\item For each variable $x_i$ in $\varphi$, there is a pair of adjacent nodes $x_i$, $\overline{x_i}$. We denote the set of these nodes~$X$.
			\item Let $s = 6y\cdot k$. Each node $x_i \in X$ is adjacent to $3sx+1$ nodes in $Z_R$ and $s(3y-1) - 3sx$ nodes in~$Z_B$.
			\item There is a clique $C$ of $m$ nodes, one node $c_i$ corresponding to each clause~$c_i$ in~$\varphi$.
			\item Each node $c_i \in C$ is adjacent to $3sx$ nodes in $Z_R$ and $s(3y-1) - 3sx$ nodes in~$Z_B$.
			\item Furthermore, each node $c_i \in C$ is adjacent to a set $Y_i$ of $k$ nodes. Let $ \bigcup_{i = 1}^m Y_i = Y$.
			\item Each node $v \in Y$ is adjacent to $3x-3$ nodes in $Z_R$ and $3(y-x)$ nodes in~$Z_B$.

			\item If $q$ is large enough, each node in $Z$ is adjacent to at most one node outside~$Z$.
		\end{itemize}
		\begin{itemize}

			\item $r = |Z_R| + k$ and $b = |Z_B|$
			\item $\sigma^{\Gamma(\varphi, \Lambda,q)}$ is a strategy profile for the game $\Gamma(\varphi,\Lambda,q)$ with
			\begin{itemize}
				\item $|Z_R|$ red agents on $Z_R$ and $k$ red agents on $Y$,
				\item $|Z_B|$ blue agents on $Z_B$.
			\end{itemize}
		\end{itemize}

	\end{definition}

	We start with a few observations that hold for the initial strategy profile $\sigma^{\Gamma(\varphi, \Lambda,q)}$ and, in fact, for any strategy profile $\sigma$ that is identical to $\sigma^{\Gamma(\varphi, \Lambda,q)}$ on all nodes in~$Z$, i.e., any profile in which all nodes in $Z_R$ are occupied by red and all nodes in~$Z_B$ are occupied by blue agents.
	\begin{itemize}
		\item An agent $i$ with $\sigma(i) \in X$ is adjacent to $3sx+1$ red agents in $Z_R$, $s(3y-1)-3sx$ blue agents in $Z_B$ and further $w \geq 0$ red agents outside $Z$ (not including herself). Therefore, $i$ has a utility of $$\p{3sx+2+w}{s(3y-1)+2+w}.$$ Observe that since for $0 < a < b$ it holds that $\tfrac{a}{b} < \tfrac{a+1}{b+1}$, we have that $$\tfrac{3sx+2+w}{s(3y-1)+2+w} > \tfrac{3sx+2}{s(3y-1)+2} > \tfrac{3sx}{3sy-1s} >  \tfrac{x}{y}.$$
		Consequently, the highest utility that agent~$i$ can obtain is $$U_{\max}^X  = \p{3sx+2}{s(3y-1)+2}.$$
		\item An agent $i$ with $\sigma(i) \in C$ is adjacent to $3sx$ red agents in~$Z_R$, $s(3y-1)-3sx$ blue agents in~$Z_B$ and further $w\geq 0$ red agents outside~$Z$ (not including herself). Hence, we have that $$U_i(\sigma) = \p{3sx+1+w}{s(3y-1)+1+w}.$$ Note that it holds that $$\tfrac{3sx+1+w}{s(3y-1)+1+w} > \tfrac{x}{y}.$$ Therefore, the utility of agent~$i$ is at most $$U_{\max}^C = \p{3sx+1}{s(3y-1)+1}.$$
		\item An agent $i$ with $\sigma(i) \in Y_a$ has $3x-3$ red neighbors in~$Z_R$ and $3(y-x)$ blue neighbors in~$Z_B$ and potentially one red neighbor on~$c_a$. Therefore, agent~$i$ has a utility of $$U_{\max}^Y = \p{3x-1}{3y-1} = \p{3x}{3y-1}$$ if~$c_a$ is occupied and $$U_{\min}^Y =\p{3x-2}{3y-2}=\p{3x}{3y-2}$$ otherwise.
		\item We claim that $$\tfrac{3sx+2}{s(3y-1)+2} < \tfrac{3x}{3y-2}.$$ In particular, it holds for $$s>\tfrac{6y-6x-4}{3 x}$$ and $$x \geq 1 \land y \geq x + 1.$$ Since we have chosen $s = 6yk > 6y$, the inequality holds. Therefore, we have that $$U_{\max}^X = \p{3sx+2}{s(3y-1)+2} > \p{3x}{3y-2} = U_{\min}^Y.$$
	\end{itemize}

	In summary, this gives the following order of utilities:
	\begin{align*}
		1 &> U_{\max}^Y = \p{3x}{3y-1}
		> U_{\max}^C = \p{3sx+1}{s(3y-1)+1}
		> U_{\max}^X = \p{3sx+2}{s(3y-1)+2}
		>  U_{\min}^Y = \p{3x}{3y-2}.
	\end{align*}

	With this, we can show that the agents that start on nodes in~$Z$ have a higher utility than they could achieve by jumping to a node outside~$Z$.
	\begin{lemma}
		Let $\sigma$ be a strategy profile that is identical to the initial placement $\sigma^{\Gamma(\varphi, \Lambda,q)}$ on all nodes in~$Z$.
		There is a $q \in \mathbf{N}$, polynomial in~$\varphi$, such that no agent on a node in~$Z$ has an improving jump.
		\label{lemma:hardness_Z_has_no_improving_jumps}
	\end{lemma}
	\begin{proof}
		Let $\sigma$ be such a strategy profile and let~$i$ be an agent with $\sigma_(i) \in Z$ that has an improving jump to an empty node~$v$.
		Since $\sigma$ is identical to $\sigma^{\Gamma(\varphi, \Lambda,q)}$ with respect to~$Z$, the empty node~$v$ must be in $X \cup C \cup Y$ and all nodes in $ X \cup C \cup Y$ must either empty or occupied by a red agent. Furthermore, with~$q$ greater than the number of edges between $X \cup C \cup Y$ and $Z$, we have that~$\sigma(i)$ can be adjacent to at most one node $u \not\in Z$.

		We begin with the observation that agent~$i$ cannot have an improving jump if~$u$ is empty, since, in that case, $$U_i(\sigma) = \p{qx}{qy-qx+qx}$$ which is  $\p{x}{y} = 1$. Hence, we can assume that~$u$ is occupied by a red agent and $v\neq u$, i.e., the empty node~$v$ is not adjacent to~$\sigma(i)$.

		We observe that agent~$i$ has a utility close but not equal to~$1$. If~$i$ is red, we get $$U_i(\sigma) =\p{qx+1}{qy-qx+qx+1} = p(\tfrac{x+1/q}{y + 1/q})$$ and if~$i$ is blue, it holds that $$U_i(\sigma)=\p{qx}{qy-qx+qx+1} = p(\tfrac{x}{y + 1/q}).$$ Note, that by choosing a sufficiently large $q$, we can have $U_i(\sigma)$ arbitrarily close to 1 and the required size of~$q$ for this is polynomial in~$\varphi$ for a fixed value of $\Lambda = \tfrac{x}{y}$. Therefore, we are left with showing that $U_i(\sigma_{iv})$ is strictly less than~$1$. It follows that $U_i(\sigma_{iv}) < U_i(\sigma)$, i.e., agent~$i$ does not have an improving jump.

		According to our previous observations, if~$i$ is red, the highest utility outside of~$Z$ it can get is $$U_{\max}^Y = \p{3x-1}{3y-1} < 1.$$

		It now remains to show that this also holds if~$i$ is a blue agent.
		If $v \in Y$, we have that either $$U_i(\sigma_{iv}) = \p{3(y-x)+1}{3y-2}$$ or $$U_i(\sigma_{iv}) = \p{3(y-x)+1}{3y-2+1},$$ depending on whether or not the corresponding node in~$C$ is occupied by a red agent. Observe that therefore, we have $U_i(\sigma_{iv})=1$ if and only if $$\tfrac{3(y-x)+1}{3y-2} = \tfrac{x}{y}$$ (resp. $\tfrac{3(y-x)+1}{3y-2+1} = \tfrac{x}{y}$). Solving for~$x$, we get that these equations hold for $$x = \tfrac{y(3y+1)}{6y-2}$$ (resp. $x = \tfrac{y(3y+1)}{6y-1}$), which we excluded by our choice of~$x$ and $y$ earlier. It follows that $U_i(\sigma_{iv}) < 1$.

		If $v \in X \union C$, $v$ is adjacent to $3sx$ or $3sx+1$ red agents in~$Z_R$ and exactly $s(3y-1)-3sx$ blue agents in~$Z_B$. With a case distinction on~$\Lambda$, we can show that this implies that $U_i(\sigma) < 1$.
		\begin{itemize}
			\item If $\Lambda \geq \tfrac{1}{2}$, i.e., $y \leq 2x$, we have that $$s(3y-1)-3sx \leq s(6x-1)-3sx = 3sx-s,$$ thus, including~$i$, there are at most $3sx-s+1 < 3sx$ blue agents in~$\neighborhood{i}{\sigma_{iv}}$. It holds that $$f_i(\sigma_{iv}) < \tfrac{1}{2} \leq \Lambda,$$ i.e., $$U_i(\sigma_{iv}) < 1.$$
			\item Otherwise, if $\Lambda < \tfrac{1}{2}$, i.e., $y \geq 2x+1$, it follows that $$s(3y-1)-3sx \geq s(6x+3-1)-3sx = 3sx+2s = 3sx+12ky.$$ Hence, we have at least $3sx+12ky$ blue agents in~$\neighborhood{i}{\sigma_{iv}}$. There can be at most $3sx+1+k$ red agents in~$\neighborhood{i}{\sigma_{iv}}$, therefore we clearly have $$f_i(\sigma_{iv}) > \tfrac{1}{2} > \Lambda.$$ Thus, it holds that $$U_i(\sigma_{iv}) < 1.$$
		\end{itemize}
		Hence, for both red and blue agents on~$Z$, we have that $U_i(\sigma_{iv}) < 1$ and thus, we can choose $q$ large enough to ensure that $i$ does not have an improving jump.
	\end{proof}
	We now get that all agents placed on~$Z$ behave like stubborn agents, i.e., do not jump at all.
	\begin{corollary}
		Starting from $\sigma^{\Gamma(\varphi, \Lambda,q)}$ and $q$ sufficiently high, every NE reached through improving response dynamics must be identical to~$\sigma_0$ on all nodes in~$Z$.
		\label{lemma:hardness_improving_response_dynamics_preserve_initial_Z}
	\end{corollary}

	\begin{lemma}
		Let $\sigma$ be a strategy profile for $\Gamma_{\varphi}$ that is identical to $\sigma^{\Gamma(\varphi, \Lambda,q)}$ on all nodes in~$Z$.
		Then,~$\sigma$ cannot be a NE, if
		\begin{enumerate}
			\item there is an agent $i$ with $\sigma(i) \in C$, or
			\item there are agents $i,j$ with $\exists l: \sigma(i) = x_l, \sigma(j) = \overline{x_l}$, or
			\item there is an agent $i$ with $\sigma(i) \in Y$.
		\end{enumerate}
		\label{lemma:hardness_JE_restrictions}
	\end{lemma}
	\begin{proof}
		We prove that under the conditions (1) -- (3), the strategy profile~$\sigma$ cannot be a NE.
		\begin{enumerate}
			\item Assume there is an agent~$i$ with $\sigma(i) = c_a \in C$.

			Consider the case that it holds for all agents~$j$ with $j \neq i$ that $\sigma(j) \notin C$. Then, all agents not on nodes in~$Z$ must be on nodes in~$Y_a$ as otherwise jumping on a node in~$Y_a$ improves their utility. Hence, agent~$i$ has a utility of $$\p{3sx+k}{s(3y-1)+k}.$$ As $k \geq 3$, this is lower than $$U_{\max}^X = \p{3sx+2}{s(3y-1)+2}.$$ Therefore,~$i$ has an incentive to jump to an arbitrary node in~$X$.

			Consider now the case that there is an agent $j\neq i$ with $\sigma(j) \in C$. Since $\abs{Y_a} = k$, there must be an empty node~$v$ in~$Y_a$. We have that $$U_j(\sigma) \leq U_{\max}^C < U_{\max}^Y = U_i(\sigma_{jv}).$$ $\sigma$ cannot be a NE.
			\item Assume that there are agents $i,j$ with $\exists l: \sigma(i) = x_l, \sigma(j) = \overline{x_l}$. If any node in~$C$ is occupied by an agent, condition~(1) shows that it cannot be a NE. So, both agent~$i$ and~$j$ have a utility of $$\p{3sx+3}{s(3y-1)+3} < U_{\max}^X .$$ Yet, by counting there must be a pair of empty nodes $x_l, \overline{x_l} \in X$, such that $U_i(\sigma_{ix_l}) = U_{\max}^X > U_i(\sigma)$.
			\item Assume there is an agent~$i$ with $\sigma(i) \in Y$. If there is an agent on a node in~$C$, then according to condition~(1), it cannot be a NE. Therefore, $U_i(\sigma) = U_{\min}^Y$. Also by counting we get that there must be a pair of empty nodes $x_l, \overline{x_l} \in X$, such that $$ U_i(\sigma_{ix_l}) = U_{\max}^X > U_i(\sigma). \qedhere$$
		\end{enumerate}
	\end{proof}

	\begin{restatable}{theorem}{theoremfive}
		For any fixed $\Lambda = \tfrac{x}{y} \in (0,1)$, it is NP-hard to decide if a given \selfInclusiveGame~played on a graph~$G$ with~$r$ red and~$b$ blue agents can reach a NE through IRDs starting from a given initial placement~$\sigma_0$.
		\label{theorem:hardness_arbitrary_lambda}
	\end{restatable}
	\begin{proof}
		Let $\varphi$ be a satisfiable instance of \textsc{Double 4-SAT} with~$k$ variables. Consider the game $\Gamma(\varphi, \Lambda,q)$ in which~$q$ is chosen sufficiently high and strategy profile~$\sigma$ is identical to $\sigma^{\Gamma(\varphi, \Lambda,q)}$ with respect to~$Z$, in which the~$k$ other red agents are placed on the nodes corresponding to the true literals of a satisfying assignment for~$\varphi$.

		We want to show that~$\sigma$ is a NE. From \Cref{lemma:hardness_improving_response_dynamics_preserve_initial_Z}, it follows that no agent~$i$ with~$\sigma(i)~\in~Z$ has an improving jump. It remains to show that no agent~$i$ with~$\sigma(i)~\in~X$ has an improving jump to some node in $X, C$ or~$Y$.

		Let~$i$ be an agent with $\sigma(i) \in X$. In~$\sigma$, all agents outside~$Z$ are on nodes in~$X$ and not adjacent to each other. Hence, it holds that $U_i(\sigma) = U_{\max}^X$.
		Since $U_i(\sigma)$ is the highest obtainable utility on nodes in~$X$, agent~$i$ cannot have an improving jump to a node in~$X$.

		In a satisfying assignment, at least two literals per clause are true. Therefore, we have that every $c_i \in C$ is adjacent to at least two nodes in~$X$ that are occupied by a red agent. Let $c_a \in C$. Thus, there is an agent $j, j \neq i, \sigma(j) \in X$ with~$\sigma(j)$ adjacent to~$c_a$, and therefore it holds that $$U_i(\sigma_{iv}) \leq \p{s3x+2}{s(3y-1)+2} = U_{\max}^X = U_i(\sigma).$$ This means that~$i$ has no improving jump to a node in~$C$.

		Furthermore, since all nodes in~$C$ are empty, all nodes in~$Y$ offer a utility of~$U_{\min}^Y~<~U_i(\sigma)$ and $\sigma$ is a NE. It can be reached with improving response dynamics from $\sigma^{\Gamma(\varphi, \Lambda,q)}$, as all~$k$ red agents outside of~$Z$ start on~$Y$ with a utility of $U_{\min}^Y$ and can, one after another, perform an improving jump to the appropriate position on~$X$.

		Let $\sigma$ be a NE for $\Gamma_\varphi$, reached through improving response dynamics starting from~$\sigma^{\Gamma(\varphi, \Lambda,q)}$. From \Cref{lemma:hardness_JE_restrictions} we get that no agent is on a node in~$Y$ or~$C$ and furthermore, for each of the~$k$ variables, exactly one literal node is occupied by one of the~$k$ red agents that did not start on~$Z$ (as no two literal nodes belonging to the same variable can both be occupied) and each of these~$k$ agents has a utility of~$U_{\max}^X$.

		Assume for the sake of contradiction that one clause node $c_i \in C$ is not adjacent to at least two red agents in~$X$. If it is adjacent to no such agent, all~$k$ red agents outside~$Z$ have an incentive to jump there as $U_{\max}^C > U_{\max}^X$. If $c_i$ is adjacent to exactly one agent~$j$ in~$X$, agent~$j$ can jump to~$c_i$ and will also have a utility of $U_{\max}^C > U_j(\sigma)$ afterward. Thus, all clause nodes are adjacent to at least two red agents and thus~$\varphi$ is a satisfiable Double 4-SAT instance.
	\end{proof}

	\subsection{{Existence of Strategy Profiles with High Degree of Integration}}
	In this section, we study the problem of finding strategy profiles with a high DoI, i.e., we aim for finding a strategy profile~with a DoI larger than some threshold~$d$. This problem is indifferent to the utilities of the agents and thus the same for any \textit{Jump Schelling Game} (JSG).
	For $d = n$, the hardness of this problem has been studied before by \cite{A+19}.
	However, their focus lies on swap games and therefore assumes $\abs{V} = n$. As noted by the authors this result can be generalized to $\abs{V} > n$ by adding isolated empty nodes. We improve on their result by showing that the hardness holds in a more realistic setting without isolated nodes.
	For our reduction, we use the NP-complete \textsc{MAX SAT} (\cite{GAREY1976SomeSimplifiedNPCompleteGraphProblems}) problem, which is defined as follows.
	\begin{definition}[\textsc{MAX SAT}]
		Given a Boolean formula $\phi$ in CNF and integer~$q$, decide if there is an assignment that satisfies at least~$q$ clauses.
	\end{definition}

We now show, that it is NP-complete to decide whether a strategy profile in which at least $d$ agents are not segregated, for some fixed $d$, exists. The proof can be found in the Appendix.
	
	\begin{restatable}{theorem}{thmfivenine}
		Given a JSG with $r$ red and $b$ blue agents on a connected graph $G=(V,E)$ with $\abs{V}>n$, it is NP-complete to decide if there is a strategy profile~$\sigma^*$ with $\text{DoI}(\sigma^*) \geq d$.
	\end{restatable}
	\begin{proof}
		Membership in NP is trivial. For a given instance~$\phi$ with a CNF consisting of~$k$ variables and~$m$ clauses and a required number of fulfilled clauses~$q$, let~$h(\phi)$ be a JSG $(G,r,b)$ and a number of non-segregated agents $d = (m+4)k+q$ where $b = k, r = (m+3) k + m$. The graph~$G$, displayed in \Cref{fig:hardness_doi_geq_d}, has one clique $X_i, i \in [k]$ with $\abs{X_i} = m+4$ for each variable~$x_i$ of~$\phi$, with two special nodes labeled~$x_i$ and~$\overline{x_i}$ in this clique corresponding to the two literals of the variable. Furthermore, for each clause~$c_i = (l_1, l_2, l_3)$ there is one node~$c_i$ connected to the nodes corresponding to~$l_1, l_2$ and~$l_3$. Let~$C$ be the set of nodes corresponding to the clauses.
		Finally, there is another node~$v$ adjacent to~$c_0$. Note that $|V| = n+1$.
		\begin{figure}[h]
		\centering
		\includegraphics[width=0.55\textwidth]{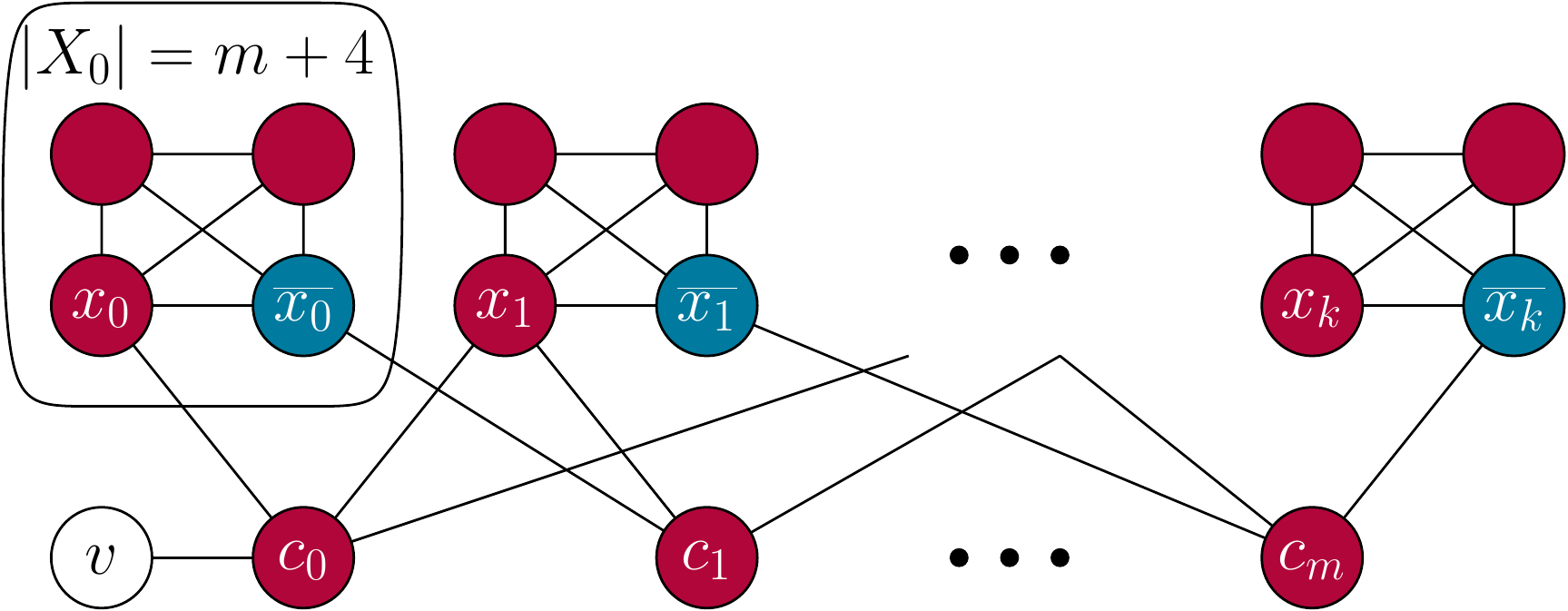}
		\caption{Construction used in the NP-hardness reduction of finding a placement $\sigma^*$ with $\text{DoI}(\sigma^*)\geq d = (m+4)k+q$. Each clique $X_i$ contains $m+4$ nodes.}
		\label{fig:hardness_doi_geq_d}
	\end{figure}

		Without loss of generality, we can assume that~$G$ is connected since we can construct an equivalent instance by adding additional, trivially satisfiable clauses that merge different connected components of $G$ to $\phi$ and increasing $q$.
		Let there be an assignment~$t$ for~$\phi$ that fulfills at least~$q$ clauses.
		Consider the placement $\sigma$, in which for all $i \in [k]$, the node $x_i$ (resp. $\overline{x_i}$) is occupied if the variable $x_i$ in $t$ is true (resp. false) and all other nodes except for $v$ are occupied by red agents. At least $(m+4)k+q$ agents are not segregated.

		Let $\sigma$ be a strategy profile with $\text{DoI}(\sigma) \geq (m+4)k+q$.
		First, observe that each $X_i$ must contain exactly one of the $k$ blue agents, as otherwise at least $\abs{X_i}-3 > m$ nodes are not segregated, contradicting a high DoI.
		Without loss of generality, we can assume that in each~$X_i$, either~$x_i$ or~$\overline{x_i}$ is occupied by the blue agent and further that $v$ is the empty node, as other configurations have a lower DoI.
		Then, from $\text{DoI}(\sigma) \geq (m+4)k+q$, it follows that at least~$q$ of the red agents on~$C$ are adjacent to a blue agent.
		Hence, an assignment in which a variable~$x_i$ is true if and only if~$\sigma(x_i)$ is a blue agent fulfills at least~$q$ clauses.
	\end{proof}

	\section{Discussion and Future Work}
	Our paper sheds light on Jump Schelling Games with non-monotone agent utilities. With this, we strengthen the recent trend of investigating more realistic residential segregation models.
	
	\subsection{Comparison with Single-Peaked Swap Schelling Games}
	Similarly to other variants of Schelling games, we also observe that our jump version behaves very differently compared to the swap version studied by \citet{BiloBLM22} and novel techniques are required. The main difference in jump games is that structural properties of the underlying graph cannot be exploited. The reason is that empty nodes are not counted when computing an agent's utility and hence it is impossible to distinguish between an empty node or a missing node. We do carry over some ideas from Single-Peaked Swap Schelling Games, e.g., the PoA upper bound proof, or the idea of considering independent sets, but the main part of our paper, e.g., all lower bound proofs and the proofs of our hardness results, follow entirely new approaches. 
	
	We obtained predominantly negative results with regard to convergence towards equilibria, in particular the finite improvement property does not hold for any $\Lambda\in (0,1)$, not even on regular graphs or trees. This is in stark contrast to the swap version, which converges to equilibria even on almost regular graphs for $\Lambda \leq \frac12$.
	Furthermore, on regular graphs with $\Lambda=\frac12$, instances of our jump version exist that do not admit equilibria. Also, although we get similar PoA bounds, compared to the swap version, we find that the PoS of the jump version tends to be worse, in particular, while the swap version has a PoS of at most $2$ on bipartite graphs for $\Lambda=\frac12$, there exists a tree that enforces a PoS that is linear in $n$ for our jump version for this setting.

	\subsection{The Variant with Self-Exclusive Neighborhoods}
	To enable a better comparison with the models by \citet{CLM18} and \citet{A+19}, that do not count the agent herself in the computation of the fraction of same-type neighbors, we also considered a variant of our model with self-exclusive neighborhoods, i.e., where the agent herself is not contained in her neighborhood. This self-exclusive variant behaves in some aspects very similarly to our model: the FIP does not hold and there is no equilibrium existence guarantee on regular graphs. Regarding the PoA it gets even worse, since equilibria exist where every agent has utility $0$, implying an unbounded PoA. This also holds for the PoA with respect to the utilitarian social welfare. Moreover, also the PoA and the PoS with respect to the utilitarian welfare is unbounded.

	\subsection{Directions for Future Work}
	We focus on measuring the social welfare of a strategy profile via the degree of integration and show that our PoA and PoS bounds also translate to PoA and PoS bounds with respect to the utilitarian social welfare. Future work could investigate these bounds in more detail, in particular, lower bounds are missing.  
	
	A main open problems for Single-Peaked Jump Schelling Games as well as for Single-Peaked Swap Schelling Games is to settle the complexity of deciding equilibrium existence. Our hardness result for finding equilibria via improving response dynamics and the observation that deciding equilibrium existence is NP-hard if stubborn agents are allowed, lead us to the following conjecture. 
	\begin{conjecture} 
	 For any peak $\Lambda \in (0,1)$ and for both the Single-Peaked Jump Schelling Game and the Single-Peaked Swap Schelling Game, it is NP-hard to decide if a given instance admits NE. 
	\end{conjecture}
	Another ambitious goal is to characterize under which conditions equilibria exist for certain graph classes. However, this is open for all known Schelling Games. 
	
	Also, it is not obvious at all how to generalize the single-peaked models to more than two agent types. As discussed by~\citet{E+19}, this is already non-trivial for the model with monotone utility functions. The simplest setting would be the "1-versus-all" variant from \citet{E+19}, where the utility only depends on the numbers of same-type and other-type neighbors. But, as shown by the authors, even in this simple setting the behavior of Schelling Games changes drastically. We expect similarly drastic changes for the single-peaked model. However, we are not convinced that "1-versus-all" captures realistic agent behavior. Ideally, in a setting with more than two types, a diverse neighborhood should contain agents of many different types and it should be balanced such that no subgroup dominates the neighborhood. 
	
	Other interesting directions for future work are further classes of realistic non-monotone utility functions. Candidates for this are plateau functions, e.g., agents have some minimum and maximum diversity requirement and are content as long these requirements are met. Or single-peaked functions that do not fulfill property (2) in our definition, like the single-peaked utilities with different slopes on both sides of the peak as used by \citet{Zha04b}. 
	
	Also, in our model we assumed that an agent explicitly considers her own contribution to the type distribution in her neighborhood. This realistic feature could also be applied to the threshold-based models by \citet{CLM18} and \citet{A+19}.
	
	\bibliographystyle{ACM-Reference-Format}
	\bibliography{sample}

\end{document}